\documentclass[11pt]{article}
\newtheorem{theorem}{Theorem}
\newtheorem{lemma}{Lemma}
\newtheorem{proof}{Proof}
\usepackage{hyperref}
\newtheorem{remark}{Remark}
\usepackage{amsmath}
\usepackage{graphicx}
\usepackage{amssymb}
\usepackage{graphicx}
\usepackage[dvipsnames,usenames]{color}
\usepackage{subfigure}
\usepackage{multirow}
\usepackage[extra]{tipa}
\usepackage{array}
\usepackage[affil-it]{authblk}
\title{On the analysis of partially homogeneous nearest-neighbour random walks in the quarter plane}

\author{Ioannis Dimitriou \footnote{E-mail: idimit@math.upatras.gr\\Website: \href{https://thalis.math.upatras.gr/~idimit/}{https://thalis.math.upatras.gr/~idimit/}}}
\affil{\small Department of Mathematics, 
University of Patras, P.O.~Box 
26500, Patras, Greece.}
\begin{document}
\maketitle
\begin{abstract} 
This work deals with the stationary analysis of two-dimensional partially homogeneous nearest-neighbour random walks. Such type of random walks in the quarter plane are characterized by the fact that the one-step transition probabilities are functions of the state-space. 
We show that its stationary behavior is investigated by solving a finite system of linear equations, and a functional equation with the aid of the theory of Riemann(-Hilbert) boundary value problems. This work is strongly motivated by emerging applications in multiple access systems as well as in the study of a general class of queueing systems with state dependent parameters. A simple numerical illustration providing useful information about a queue-aware multiple access system is also presented. \vspace{2mm}\\
{\textbf{Keywords:}} State-dependency, Nearest-neighbour random walk, Stationary distribution, Boundary value problem.
\end{abstract}
\section{Introduction}
In this work we focus on the stationary analysis of irreducible discrete time Markov chains in the quarter plane $\mathbb{Z}_{+}^{2}$ (where $\mathbb{Z}_{+}$ refers to the set of non-negative integers), whose one-step transition probabilities possess a partial homogeneity property. More precisely, we focus on nearest-neighbour two-dimensional random walk with one-step transition probabilities defined as follows: transitions from an interior point $(n_{1},n_{2})\in\{1,2,\ldots\}\times\{1,2,\ldots\}$ of the state space lead with probability $p_{i,j}(n_{1},n_{2})$ to a neighbouring point $(n_{1}+i,n_{2}+j)$, where $(i,j)\in\{-1,0,1\}\times\{-1,0,1\}$. 

Such a class of state-dependent two-dimensional random walks are instrumental in the analytical investigation of a large class of queueing networks with \textit{interacting queues}, where interaction means that system parameters are functions of the state of the network. However, the stationary analysis of general state-dependent two-dimensional random walks is still an open problem. 

In this paper, we focus on the partial-homogeneous nearest neighbour random walks in the quarter plane (PH-NNRWQP; see Figure \ref{fig10}), which obey the following property: The state space $S=\{(n_{1},n_{2});n_{1},n_{2}\geq0\}$ is split in four non-intersecting subsets, i.e., $S=S_{0}\cup S_{1}\cup S_{2}\cup S_{3}$, where:
\begin{equation}
\begin{array}{rl}
S_{0}=\{(n_{1},n_{2});n_{1}<N_{1},n_{2}<N_{2}\},&S_{1}=\{(n_{1},n_{2});n_{1}\geq N_{1},n_{2}<N_{2}\},\\
S_{2}=\{(n_{1},n_{2});n_{1}<N_{1},n_{2}\geq N_{2}\},&S_{3}=\{(n_{1},n_{2});n_{1}\geq N_{1},n_{2}\geq N_{2}\},
\end{array}\label{split}
\end{equation}
such that for $i,j=0,\pm 1$
\begin{equation}
\begin{array}{rl}
p_{i,j}(n_{1},n_{2})=&\left\{\begin{array}{ll}
p_{i,j}(N_{1},n_{2}),&(n_{1},n_{2})\in S_{1},\\
p_{i,j}(n_{1},N_{2}),&(n_{1},n_{2})\in S_{2},\\
p_{i,j}(N_{1},N_{2}):=p_{i,j},&(n_{1},n_{2})\in S_{3}.
\end{array}\right.
\end{array}\label{trans}
\end{equation}
Our aim in this work is to provide an analytical approach to investigate its stationary behavior, and state the importance of the models described by such class of RWQP in the modelling (among others) of emerging engineering applications in multiple access systems.
\begin{figure}[htp]
\centering
\includegraphics[scale=0.5]{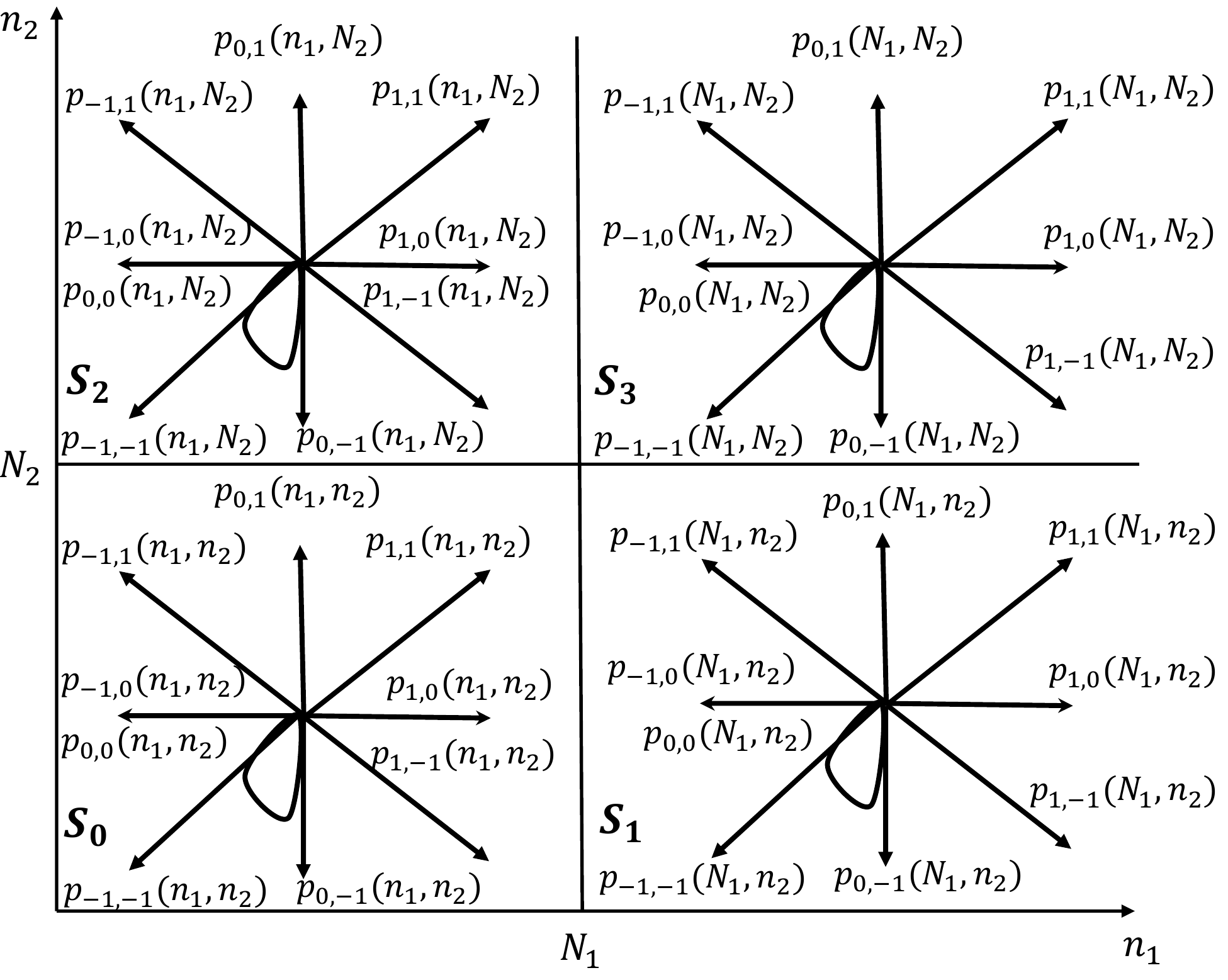}
\caption{The PH-NNRWQP.}
\label{fig10}
\end{figure}
\subsection{Related work}
Since the pioneered works \cite{fay2,maly3}, RWQP has been extensively studied as an important topic of applied probability with links, among others in queueing theory, e.g., \cite{cof,fayolleG,fay2,fay1,coh2,coh1,flat,fphd,iphd}, in finance \cite{cont}, in combinatorics, e.g., \cite{Kurkova2015,BM1,KR1,FR,adRL}. Substantial work has also done in obtaining exact tail asymptotics, see e.g., \cite{miya,guileu,flatha,lizha1,lizha2,lizha3,oza,kurksuh,malyas} (not exhaustive list).

The main body of the related literature is devoted to the analysis of \textit{semi-homogeneous} NNRWQP. With the term \textit{semi-homogeneous}, we mean that the transition probabilities are state-independent in so far it concerns states belonging to the interior, i.e., $\{(n_{1},n_{2}),n_{1},n_{2}>0\}$, similarly for those of $\{(n_{1},n_{2}),n_{1}>0,n_{2}=0\}$, of $\{(n_{1},n_{2}),n_{1}=0,n_{2}>0\}$ and of $\{(0,0)\}$. Most of the research refers to the investigation of ergodicity conditions; e.g., \cite{rw,fay}. The derivation of the stationary performance metrics reveals is not an easy task and was performed with the aid of the theory of boundary value problems \cite{fay1,coh1}. 

Explicit conditions for recurrence and transience were given in the seminal works in \cite{maly1,malmen,men}. For a detailed treatment see the seminal book in \cite{fay}. We also refer \cite{rose} that partly extended Malyshev's work. Ergodicity conditions for the \textit{partially homogeneous} case (see Fig. \ref{fig10}) described above was considered in \cite{malmen} under the assumption that the jumps of the RWQP are bounded. It was later considered in \cite{fayo} under a weaker restriction that the jumps of the RWQP have bounded second moments; see also \cite{fay}. A profound study concerning necessary and sufficient conditions for ergodicity of general RWQP that are continuous to the \textit{West}, to the \textit{South-West} and to \textit{South} are given in \cite[Part II]{rw}. For a detailed methodological treatment of the stationary analysis of semi-homogeneous RWQP the reader is referred to the seminal books in \cite{fay1,coh1}. In \cite{fay1}, the analysis is concentrated mainly to nearest neighbour RWQPs, while in \cite{coh1} the authors provided a systematic study of RWQP that are continuous to the \textit{West}, to the \textit{South-West} and to \textit{South}. 

Among the class of NNRWQP, a very effective analytical approach can be applied when transitions to the \textit{North}, to the \textit{North-East} and to the \textit{East} are not allowed. In particular, it is shown that the bivariate generating function of the stationary joint distribution of the random walk can be explicitly expressed in terms of meromorphic functions \cite{c1,c2,c3,Adan2016}. 

The analysis becomes quite harder when space homogeneity property collapses. Such a situation arises in the two-dimensional join the shortest queue problem, in which the quarter plane is separated into two homogeneous regions \cite{fphd,iphd,coh1}. In these studies, the analysis is reduced to the simultaneous solution of two boundary value problems. In \cite{king}, the author considered the symmetric shortest queue problem and provided an analytic method to obtained its stationary distribution. A very efficient method to obtain the stationary joint queue length distribution in the shortest queue problem was developed in \cite{Adan2016,ad1,ad0,ad2,ad3} (the compensation approach), by appropriately transformed the original RWQP to a random walk with no transitions to the East, to the North-East and to North. Such an approach provides an explicit characterization of the equilibrium probabilities as an infinite or finite series of product forms.

In \cite{fayo}, the authors considered a two dimensional birth-death process (i.e., transitions were allowed only to the North, East, West and South) with partial homogeneity that separate the state space in four distinct regions. They showed that its stationary behavior is investigated by solving a boundary value problem and a system of linear equations. Stability conditions for such type of random walks was investigated in \cite{za,faysta}. 
\subsection{Our contribution}
\paragraph{Fundamental contribution}
In this work we present an analytical method for analyzing the stationary behavior of a partially homogeneous nearest-neighbour random walk in the quarter plane whose one step transition probabilities for $(n_{1},n_{2})\in S=S_{0}\cup S_{1}\cup S_{2}\cup S_{3}$, and $i,j=0,\pm 1$ are as in (\ref{trans}).
\begin{itemize}
\item We show that the determination of the steady-state distribution of a PH-NNRWQP can be reduced to the solution of a finite system of linear equations, as well as to the solution of a non-homogeneous Riemann boundary value problem. \item For the special case of a PH-NNRWQP with no transitions to South-West, we present a slightly different approach by using the theory of Riemann-Hilbert boundary value problems.
\end{itemize}

\paragraph{Applications} This general class of PH-NNRWQP serves as a general modeling framework for several engineering applications. In particular, it can be used to model interacting queues in multiple access networks.

Consider a network consisting of two users communicating in random access manner with a common destination node; see Fig. \ref{fig1w}. The time is assumed to be slotted, and each user has external arrivals that depend on the state of the network at the beginning of a slot. Arriving packets are stored in their infinite capacity queues. At the beginning of a slot, users accessing the wireless channel randomly by adapting their transmission probabilities based on the status of the network. 

Therefore, each user node adapts its transmission parameters according to its own state as well as the state of the other node. In such a case, we also take into account both the wireless interference, and the complex interdependence among users' nodes due to the shared medium \cite{abi,borst,dim1,dimpaptwc,adhoc}. Clearly, such a protocol leads to substantial performance gains, since allow to dynamically design random access networks, or equivalently allow to investigate \textit{intelligent} and \textit{self-aware} multiple access systems. 

Note that in such a shared access network, it is practical to assume a minimum exchanging information of one bit between the nodes, which allows them to be aware of the state of each other. We assume that the time for the exchange of information is very small with respect to slot duration, and thus considered negligible. Moreover, the knowledge of the state of the network by a node, does not mean that a node will remain silent if the other node is active (i.e., there are buffered packets in its queue). In particular, by allowing the nodes parameters (i.e., packet generation probability, packet transmission probability) to depend on the state of the network, we provide the additional flexibility towards self-aware and dynamically adapted networks. To the best of our knowledge this variation of random access has not been reported in the literature.

Note also that in \cite{gu,shn1,shn2,stol2}, dynamic, queue-length based strategies were introduced in order to investigate the stability of queue-aware multiple access networks. In these works, the actual queue lengths of the flows in each node's close
neighbourhood are used to determine the nodes' channel access
probabilities. However, they did to focus on the stationary behaviour. We also refer to \cite{fay,gha,borst,mal1,za,shn3}, in which stability condition for Markov chains both in two and higher dimensions, whose transition structure possess a property of spatial homogeneity was investigated. 



The paper is organized as follows. In Section \ref{mod} we present in detail the mathematical model, while in Section \ref{pre} we provide a detailed analysis about how to investigate the stationary behavior of a general PH-NNRWQP, by solving a finite system of linear equations and a functional equations in terms of a solution of a non-homogeneous Riemann boundary value problem. In Section \ref{fayo}, we focus on the case where transitions are not allowed in South-West, and show how this special case of PH-NNRWQP is analyzed with the aid of the theory of Riemann-Hilbert boundary value problems. An application on the modelling of adaptive multiple access systems is given in Section \ref{appl}, while a simple numerical illustration is given in Section \ref{num}.
\section{Model description}\label{mod}
We consider a partially homogeneous two-dimensional stochastic processes $\mathbf{Q}_{n}=\{(Q_{1,n},Q_{2,n});n=0,1,\ldots\}$, with state space $S=\mathbb{N}_{0}\times\mathbb{N}_{0}=\{0,1,\ldots\}\times\{0,1,\ldots\}$. For the complete description of the structure of $\mathbf{Q}_{n}$ we need the following assumptions:
\begin{itemize}
\item There exist two positive constants, say $N_{1}$, $N_{2}$, such that $S$ is written as $S=S_{0}\cup S_{1}\cup S_{2}\cup S_{3}$, where the non-intersecting sets $S_{j}$, $j=0,1,2,3$ are given in (\ref{split}).
\item For $j=0,1,2,3$, denote the sequence of independent stochastic vectors $$\{(\xi_{1n}^{(j)}(Q_{1,n},Q_{2,n}),\xi_{2n}^{(j)}(Q_{1,n},Q_{2,n})),n=0,1,\ldots\},$$ with range space $\{-1,0,1\}\times\{-1,0,1\}$. The distribution of the stochastic vectors depend on the state of $\mathbf{Q}_{n}$ according to the state-space splitting as shown in Figure \ref{fig10}.
\item The family $\{(\xi_{1n}^{(3)}(Q_{1,n},Q_{2,n}),\xi_{2n}^{(3)}(Q_{1,n},Q_{2,n})),n=0,1,\ldots\}$ is a sequence of i.i.d. stochastic vectors. Moreover, $(\xi_{1n}^{(3)}(Q_{1,n},Q_{2,n}),\xi_{2n}^{(3)}(Q_{1,n},Q_{2,n}))\sim(\xi_{1n}^{(3)},\xi_{2n}^{(3)})$, for $(Q_{1,n},Q_{2,n})\in S_{3}$. 
\item The four families $\{(\xi_{1n}^{(j)}(Q_{1,n},Q_{2,n}),\xi_{2n}^{(j)}(Q_{1,n},Q_{2,n}))\}$, $j=0,1,2,3$ are independent and $$(\xi_{1n}^{(j)}(Q_{1,n},Q_{2,n}),\xi_{2n}^{(j)}(Q_{1,n},Q_{2,n}))\sim(\xi_{1}^{(j)}(Q_{1},Q_{2}),\xi_{2}^{(j)}(Q_{1},Q_{2})).$$
\end{itemize}
Then, for $n=0,1,\ldots$, $k=1,2,$ and $\mathbf{Q}_{0}=(0,0)$
\begin{displaymath}
Q_{k,n+1}=[Q_{k,n}+\xi_{n}^{(j)}(Q_{1,n},Q_{2,n})]^{+},
\end{displaymath}
where $[a]^{+}=max(0,a)$.

The model at hand is described by a two-dimensional Markov chain with limited state dependency, or equivalently with partial spatial homogeneity. Conditions for ergodicity for such random walks in the positive quadrant has been investigated in \cite[Theorem 3.1, p. 178]{fay2}, \cite[Theorem 4]{za}\footnote{In Section \ref{appl} we provide the stability conditions for a certain application of a PH-NNRWQP in adaptive ALOHA-type random access networks.}.
\section{Analysis}\label{pre}
Assume hereon that the system is stable, and let the equilibrium probabilities
\begin{displaymath}
\pi(n_{1},n_{2})=\lim_{m\to\infty}P(Q_{1,m}=n_{1},Q_{2,m}=n_{2}).
\end{displaymath}
Then, for $(n_{1},n_{2})\in S$ the equilibrium equations reads
\begin{equation}
\begin{array}{l}
\pi(n_{1},n_{2})=\pi(n_{1},n_{2})p_{0,0}(n_{1},n_{2})+\pi(n_{1},n_{2}+1)p_{0,-1}(n_{1},n_{2}+1)\\
+\pi(n_{1}+1,n_{2})p_{-1,0}(n_{1}+1,n_{2})+\pi(n_{1}-1,n_{2}+1)p_{1,-1}(n_{1}-1,n_{2}+1)\mathbf{1}_{\{n_{1}\geq1\}}\\
+\pi(n_{1}+1,n_{2}-1)p_{-1,1}(n_{1}+1,n_{2}-1)\mathbf{1}_{\{n_{2}\geq1\}}+\pi(n_{1},n_{2}-1)p_{0,1}(n_{1},n_{2}-1)\mathbf{1}_{\{n_{2}\geq1\}}\\
+\pi(n_{1}-1,n_{2}-1)p_{1,1}(n_{1}-1,n_{2}-1)\mathbf{1}_{\{n_{1},n_{2}\geq1\}}\\+\pi(n_{1}+1,n_{2}+1)p_{-1,-1}(n_{1}+1,n_{2}+1),
\end{array}\label{e1}
\end{equation}
where $\sum_{n_{1}=0}^{\infty}\sum_{n_{2}=0}^{\infty}\pi(n_{1},n_{2})=1$, $\pi(n_{1},-1)=0=\pi(-1,n_{2})$ and $\mathbf{1}_{\{A\}}$ the indicator function of the event $A$. 
\subsection{Generating functions and the functional equation}
To proceed, we focus on the equilibrium equations (\ref{e1}) at each sub-region of the state space separately.
\begin{enumerate}
\item \textbf{Region $S_{0}$:} Consider first the equilibrium equations (\ref{e1}) corresponding to the region $S_{0}$. There are $N_{1}\times N_{2}$ equations ($n_{1}=0,1,...,N_{1}-1$, $n_{2}=0,1,...,N_{2}-1$) involving $(N_{1}+1)\times(N_{2}+1)$ unknown probabilities ($\pi(n_{1},n_{2})$, $n_{1}=0,1,...,N_{1}$, $n_{2}=0,1,...,N_{2}$). This leaves $N_{1}+N_{2}+1$ unknowns.
\item \textbf{Region $S_{1}$:} In the following, we focus on the equations associated with $S_{1}$ ($n_{1}=N_{1},N_{1}+1,...$, $n_{2}=0,1,...N_{2}-1$). Let,
\begin{displaymath}
g_{n_{2}}(x)=\sum_{n_{1}=N_{1}}^{\infty}\pi_{n_{1},n_{2}}x^{n_{1}-N_{1}},\,n_{2}=0,1,....
\end{displaymath}
Having in mind that $p_{i,j}(n_{1},n_{2})=p_{i,j}(N_{1},n_{2})$ for $n_{1}\geq N_{1}$, we obtain from (\ref{e1}) the following relations,
\begin{equation}
\begin{cases}\begin{array}{l}
f_{2}(N_{1},0,x)g_{0}(x)-f_{3}(N_{1},1,x)g_{1}(x)=b_{0}(x),\vspace{2mm}\\
-f_{1}(N_{1},n_{2}-1,x)g_{n_{2}-1}(x)+f_{2}(N_{1},n_{2},x)g_{n_{2}}(x)\\-f_{3}(N_{1},n_{2}+1,x)g_{n_{2}+1}(x)=b_{n_{2}}(x),\,n_{2}=1,2,...,
\end{array}\end{cases}
\label{r1}
\end{equation} 
where, for $n_{2}=0,1,2,...$,
\begin{displaymath}
\begin{array}{rl}
f_{1}(N_{1},n_{2},x)=&x^{2}p_{1,1}(N_{1},n_{2})+xp_{0,1}(N_{1},n_{2})+p_{-1,1}(N_{1},n_{2}),\vspace{2mm}\\
f_{2}(N_{1},n_{2},x)=&x[1-p_{0,0}(N_{1},n_{2})]-x^{2}p_{1,0}(N_{1},n_{2})-p_{-1,0}(N_{1},n_{2}),\vspace{2mm}\\
f_{3}(N_{1},n_{2},x)=&p_{-1,-1}(N_{1},n_{2})+xp_{0,-1}(N_{1},n_{2})+x^{2}p_{1,-1}(N_{1},n_{2}),\end{array}
\end{displaymath}
\begin{displaymath}
\begin{array}{l}
b_{n_{2}}(x)=x[\pi(N_{1}-1,n_{2}-1)p_{1,1}(N_{1}-1,n_{2}-1)\\+\pi(N_{1}-1,n_{2}+1)p_{1,-1}(N_{1}-1,n_{2}+1)
+\pi(N_{1}-1,n_{2})p_{1,0}(N_{1}-1,n_{2})]\vspace{2mm}\\-[\pi(N_{1},n_{2}-1)p_{-1,1}(N_{1},n_{2}-1)+\pi(N_{1},n_{2})p_{-1,0}(N_{1},n_{2})\vspace{2mm}\\+\pi(N_{1},n_{2}+1)p_{-1,-1}(N_{1},n_{2}+1)].
\end{array}
\end{displaymath}

Relations (\ref{r1}) allow to express $g_{n_{2}}(x)$, $n_{2}=1,2,...$, in terms of $g_{0}(x)$ and $b_{0}(x),...,b_{n_{2}-1}(x)$. Indeed, starting from the first in (\ref{r1}) and solving recursively, we conclude that,
\begin{equation}
\begin{array}{c}
g_{n_{2}}(x)=e_{n_{2}}(x)g_{0}(x)+t_{n_{2}}(x),\,n_{2}=1,2,...,
\end{array}\label{r2}
\end{equation}
where, for $n_{2}=1,2,...,$
\begin{displaymath}
\begin{array}{rl}
e_{n_{2}}(x)=&\frac{f_{2}(N_{1},n_{2}-1,x)e_{n_{2}-1}(x)-f_{1}(N_{1},n_{2}-2,x)e_{n_{2}-2}(x)}{f_{3}(N_{1},n_{2},x)},\vspace{2mm}\\
t_{n_{2}}(x)=&\frac{f_{2}(N_{1},n_{2}-1,x)t_{n_{2}-1}(x)-f_{1}(N_{1},n_{2}-2,x)t_{n_{2}-2}(x)-b_{n_{2}-1}(x)}{f_{3}(N_{1},n_{2},x)},
\end{array}
\end{displaymath}
where $e_{-1}(x)=0=t_{-1}(x)=t_{0}(x)$ and $e_{0}(x)=1$. Note that up to, and including $n_{2}=N_{2}$, no other new probabilities appear, except those introduced in the equations for the region $S_{0}$, i.e., $\pi(n_{1},n_{2})$, $n_{1}=0,1,...,N_{1}$, $n_{2}=0,1,...,N_{2}$. 

For $n_{2}=1,\ldots,N_{2}$, if 
\begin{displaymath}
\begin{array}{rl}
\textbf{l}(x):=&(g_{1}(x),\ldots,g_{N_{2}}(x))^{\prime},\\\textbf{b}(x):=&(b_{0}(x),\ldots,b_{N_{2}-1}(x))^{\prime},\\\textbf{c}_{1}(x):=&(-f_{2}(N_{1},0,x),f_{1}(N_{1},0,x),0,\ldots,0)^{\prime},
\end{array}
\end{displaymath}
the system (\ref{r1}) is written in matrix form as
\begin{equation}
\textbf{K}(x)\textbf{l}(x)=\textbf{c}_{1}(x)g_{0}(x)+\textbf{b}(x),
\label{sys1}
\end{equation}
where $\textbf{K}(x):=(k_{i,j}(x))$ is a $N_{2}\times N_{2}$ matrix with elements
\begin{displaymath}
\begin{array}{rl}
k_{i,j}(x)=&\left\{\begin{array}{ll}
-f_{3}(N_{1},i,x),&i=j,\\
f_{2}(N_{1},i-1,x),&i=j+1,\\
-f_{1}(N_{1},j,x),&i=j+2,\\
\end{array}\right.
\end{array}
\end{displaymath}

Note that system (\ref{sys1}) is non-singular since $\textbf{K}_{1}(x)$ is lower triangular, having determinant equal to $det(\textbf{K}(x))=(-1)^{N_{2}}\prod_{n_{2}=1}^{N_{2}}f_{3}(N_{1},n_{2},x)$.
\item \textbf{Region $S_{2}$:} Clearly, region $S_{2}$ is a mirror image of $S_{1}$, where index 1 becomes 2 and component $n_{1}$ becomes $n_{2}$. Similarly, denote,
\begin{displaymath}
h_{n_{1}}(y)=\sum_{n_{2}=N_{2}}^{\infty}\pi_{n_{1},n_{2}}y^{n_{2}-N_{2}},\,n_{1}=0,1,....
\end{displaymath}
By repeating the procedure, 
\begin{equation}
\begin{cases}\begin{array}{l}
\tilde{f}_{2}(0,N_{2},y)h_{0}(y)-\tilde{f}_{3}(1,N_{2},y)h_{1}(y)=u_{0}(y),\vspace{2mm}\\
-\tilde{f}_{1}(n_{1}-1,N_{2},y)h_{n_{1}-1}(y)+\tilde{f}_{2}(n_{1},N_{2}y)h_{n_{1}}(y)\\-\tilde{f}_{3}(n_{1}+1,N_{2}y)h_{n_{1}+1}(y)=u_{n_{2}}(y),\,n_{1}=1,2,...,
\end{array}\end{cases}
\label{r11}
\end{equation} 
where, for $n_{1}=0,1,2,...$,
\begin{displaymath}
\begin{array}{rl}
\tilde{f}_{1}(n_{1},N_{2},y)=&y^{2}p_{1,1}(n_{1},N_{2})+yp_{1,0}(n_{1},N_{2})+p_{1,-1}(n_{1},N_{2}),\vspace{2mm}\\
\tilde{f}_{2}(n_{1},N_{2},y)=&y[1-p_{0,0}(n_{1},N_{2})]-y^{2}p_{0,1}(n_{1},N_{2})-p_{0,-1}(n_{1},N_{2}),\vspace{2mm}\\
\tilde{f}_{3}(n_{1},N_{2},y)=&y^{2}p_{-1,1}(n_{1},N_{2})+yp_{-1,0}(n_{1},N_{2})+p_{-1,-1}(n_{1},N_{2}),
\end{array}
\end{displaymath}
\begin{displaymath}
\begin{array}{l}
u_{n_{1}}(y)=y[\pi(n_{1}-1,N_{2}-1)p_{1,1}(n_{1}-1,N_{2}-1)\\+\pi(n_{1}+1,N_{2}-1)p_{-1,1}(n_{1}+1,N_{2}-1)
+\pi(n_{1},N_{2}-1)p_{0,1}(n_{1},N_{2}-1)]\vspace{2mm}\\-[\pi(n_{1}-1,N_{2})p_{1,-1}(n_{1}-1,N_{2})+\pi(n_{1},N_{2})p_{0,-1}(n_{1},N_{2})\vspace{2mm}\\+\pi(n_{1}+1,N_{2})p_{-1,-1}(n_{1}+1,N_{2})].
\end{array}
\end{displaymath}
Then, 
\begin{equation}
h_{n_{1}}(y)=\tilde{e}_{n_{1}}(y)h_{0}(y)+\tilde{t}_{n_{1}}(y),\,n_{1}=1,2,...,
\label{r3}
\end{equation}
where $\tilde{e}_{n_{1}}(y)$, are known polynomials and $\tilde{t}_{n_{1}}(y)$ contain unknown probabilities, but now new terms except those introduced in the equations for $S_{0}$, i.e., $\pi(n_{1},n_{2})$, $n_{1}=0,1,...,N_{1}$, $n_{2}=0,1,...,N_{2}$. Similarly, (\ref{r11}) is written for $n_{1}=1,\ldots,N_{1}$ as 
\begin{equation}
\textbf{M}(y)\textbf{j}(y)=\textbf{c}_{2}(y)h_{0}(x)+\textbf{u}(y),
\label{sys2}
\end{equation}
where $\textbf{M}(y):=(m_{i,j}(y))$ is a $N_{1}\times N_{1}$ matrix with elements
\begin{displaymath}
\begin{array}{rl}
m_{i,j}(x)=&\left\{\begin{array}{ll}
-\tilde{f}_{3}(i,N_{2},y),&i=j,\\
\tilde{f}_{2}(j-1,N_{2},y),&j=i+1,\\
-\tilde{f}_{1}(j,N_{2},x),&j=i+2,\\
\end{array}\right.
\end{array}
\end{displaymath}
and 
\begin{displaymath}
\begin{array}{rl}
\textbf{j}(y):=&(h_{1}(y),\ldots,h_{N_{1}}(y))^{\prime},\\\textbf{u}(y):=&(u_{0}(y),\ldots,u_{N_{1}-1}(y))^{\prime},\\\textbf{c}_{2}(y):=&(-\tilde{f}_{2}(0,N_{2},y),\tilde{f}_{1}(0,N_{2},y),0,\ldots,0)^{\prime}.
\end{array}
\end{displaymath}
\item \textbf{Region $S_{3}$:} We now focus on the region $S_{3}$ and denote,
\begin{displaymath}
\begin{array}{rl}
g(x,y)=&\sum_{n_{1}=N_{2}}^{\infty}\sum_{n_{2}=N_{2}}^{\infty}\pi_{n_{1},n_{2}}x^{n_{1}-N_{1}}y^{n_{2}-N_{2}}\vspace{2mm}\\
=&\sum_{n_{2}=N_{2}}^{\infty}g_{n_{2}}(x)y^{n_{2}-N_{2}}=\sum_{n_{1}=N_{1}}^{\infty}h_{n_{1}}(y)x^{n_{1}-N_{1}}.
\end{array}
\end{displaymath}
Using (\ref{r1}), noting that $f_{i}(N_{1},n_{2},x):=f_{i}(N_{1},N_{2},x)$ for $n_{2}\geq N_{2}$, and having in mind (\ref{r2}), (\ref{r3}), we finally obtain after lengthy calculations,
\begin{equation}
R(x,y)g(x,y)=A(x,y)g_{0}(x)+B(x,y)h_{0}(y)+C(x,y),
\label{fun}
\end{equation}
where, 
\begin{equation}
R(x,y)=xy-\Psi(x,y),\label{kern}
\end{equation}
and,
\begin{displaymath}
\begin{array}{rl}
\Psi(x,y)=&xyp_{0,0}+x^{2}yp_{1,0}+yp_{-1,0}+x^{2}y^{2}p_{1,1}+xy^{2}p_{0,1}\\&+y^{2}p_{-1,1}+xp_{0,-1}+x^{2}p_{1,-1}+p_{-1,-1},\\
A(x,y)=&yf_{1}(N_{1},N_{2}-1,x)e_{N_{2}-1}(x)-f_{3}(N_{1},N_{2},x)e_{N_{2}}(x),\vspace{2mm}\\
B(x,y)=&x\tilde{f}_{1}(N_{1}-1,N_{2},y)\tilde{e}_{N_{1}-1}(y)-\tilde{f}_{3}(N_{1},N_{2},y)\tilde{e}_{N_{1}}(y),\vspace{2mm}\\
C(x,y)=&K(\pi(N_{1}-1,N_{2}-1),\pi(N_{1}-1,N_{2}),\pi(N_{1},N_{2}-1),x,y)\\&+yf_{1}(N_{1},N_{2}-1,x)t_{N_{2}-1}(x)-f_{3}(N_{1},N_{2},x)t_{N_{2}}(x)\\&+x\tilde{f}_{1}(N_{1}-1,N_{2},y)\tilde{t}_{N_{1}-1}(y)-\tilde{f}_{3}(N_{1},N_{2},y)\tilde{t}_{N_{1}}(y),
\end{array}
\end{displaymath}
\begin{displaymath}
\begin{array}{l}
K(\pi(N_{1}-1,N_{2}-1),\pi(N_{1}-1,N_{2}),\pi(N_{1},N_{2}-1),\pi(N_{1},N_{2}),x,y)\\
=xy\pi(N_{1}-1,N_{2}-1)p_{1,1}(N_{1}-1,N_{2}-1)+p_{-1,-1}\pi(N_{1},N_{2})\\-y\pi(N_{1},N_{2}-1)p_{1,-1}(N_{1}-1,N_{2})
-x\pi(N_{1}-1,N_{2})p_{-1,1}(N_{1},N_{2}-1).
\end{array}
\end{displaymath}
\end{enumerate}
Note that $g(x,y)$ is for every fixed $x$ with $|x|\leq 1$, regular in $y$ for $|y|<1$, continuous in $y$ for $|y|\leq 1$, and similarly with $x$, $y$ interchanged.
\begin{remark} Note that for the general case we have $\Psi(0,0)>0$. The analysis is considerably different when $\Psi(0,0)=0$, i.e., when $p_{-1,-1}=0$; see Section \ref{fayo} for more details.\end{remark}
 
Note that for the probabilities of states in region $S_{0}$ we have equations (\ref{e1}), for  $n_{1}= 0, 1,..., N_{1}-1$, $n_{2} = 0, 1,..., N_{2}-1$. For those in region $S_{1}$, we have equations (\ref{r2}), $n_{2}=0, 1,...,N_{2} -1$; for those in $S_2$, we have equations (\ref{r3}), $n_{1} = 0, 1,..., N_{1}-1$.  
For region $S_{3}$, all unknown quantities are expressed in terms of
\begin{enumerate}
\item  $g_{0}(x)$, $h_{0}(y)$, 
\item the $N_{1}+N_{2}+1$ probabilities, $\pi(N_{1},n_{2})$, $n_{2} =0, 1,...,N_{2} -1$, and $\pi(n_{1},N_{2})$, $n_{1} =0, 1,...,N_{1} -1$ and $\pi(N_{1},N_{2})$.  
\end{enumerate}

\subsection{Kernel analysis}
Our aim in this section is to determine $g_{0}(x)$, $h_{0}(y)$, in terms of the solution of a Riemann boundary value problem. Thus, as a first step, we have to investigate the zeros of the kernel equation $R(x,y)=0$. 

Consider the kernel for 
\begin{equation}
x=gs,\,\,\,y=gs^{-1},\,\,|s|=1,\,|g|\leq1.\label{dgh}
\end{equation}
It follows that for $s=e^{i\phi}$, (\ref{dgh}) define a one-to-one mapping $f$ between $x(\phi)$ and $y(\phi)$, i.e., $x(\phi)=f(y(\phi))$ or $y(\phi)=f^{-1}(x(\phi))$, $\phi\in[0,2\pi)$, and
\begin{equation}\\
\begin{array}{rl}
R(gs,gs^{-1})=&g^{2}-\Psi(gs,gs^{-1})=g^{2}-\mathbb{E}(g^{\xi_{1}^{(3)}+\xi_{2}^{(3)}}s^{\xi_{1}^{(3)}-\xi_{2}^{(3)}})\\
=&g^{2}-\sum_{i=-1}^{1}\sum_{j=-1}^{1}p_{i,j}g^{i+j+2}s^{i-j}.
\end{array}
\end{equation}
It is readily seen that $\Psi(gs,gs^{-1})$ is for $g=1$ regular at $s=1$, and for $s=1$, regular at $g=1$ (i.e., all moments exist and are finite), whereas,
\begin{equation}
\begin{array}{l}
R(gs,gs^{-1})=0\Leftrightarrow g^{2}=\frac{p_{-1,-1}+p_{-1,0}gs^{-1}+p_{0,-1}gs}{1-\sum_{\substack{ i=-1\\i+j\geq0}}^{1}\sum_{j=-1}^{1}p_{i,j}g^{i+j}s^{i-j}}.
\end{array}\label{fty}
\end{equation}
Note that for $|g|\leq1$, $|s|=1$, the denominator in (\ref{fty}) never vanishes. Indeed, 
\begin{displaymath}
\begin{array}{rl}
|\sum_{\substack{ i=-1\\i+j\geq0}}^{1}\sum_{j=-1}^{1}p_{i,j}g^{i+j}s^{i-j}|\leq& p_{0,0}+p_{0,1}+p_{1,-1}+p_{1,0}+p_{1,1}+p_{-1,1}\\=&1-(p_{-1,-1}+p_{-1,0}+p_{0,-1})<1.
\end{array}
\end{displaymath}

Let $E_{x}=p_{1,0}+p_{1,1}+p_{1,-1}-(p_{-1,1}+p_{0,-1}+p_{-1,-1})$, $E_{y}=p_{0,1}+p_{1,1}+p_{-1,1}-(p_{-1,-1}+p_{0,-1}+p_{1,-1})$, i.e., the mean drifts in region $S_{3}$. 
\begin{theorem}\label{th1}\begin{enumerate}
\item If $E_{x}<0$, $E_{y}<0$, the kernel $R(gs,gs^{-1})$, $|s|=1$ has in $|g|\leq1$ exactly two zeros each with multiplicity one, which are both real for $s=\pm 1$.
\item If $g(s)$ is a zero, so is $-g(-s)$.
\end{enumerate}
\end{theorem}
\begin{proof}
See Appendix \ref{nnrw}.
\end{proof}

Define,
\begin{displaymath}
\mathcal{S}_{1}:=\{x:x=g(s)s,|s|=1\},\,\,\mathcal{S}_{2}:=\{y:y=g(s)s^{-1},|s|=1\},
\end{displaymath}
where $g(s)$ the positive zero of the kernel. In the following we have to show that $S_{1}$, $S_{2}$ are simple and smooth, i.e., they are closed, non-self intersecting curves with a continuously varying tangent; see Figures \ref{conto}, \ref{conto1} for some values of the parameters.
\begin{figure}[ht!]
\centering
\includegraphics[scale=0.6]{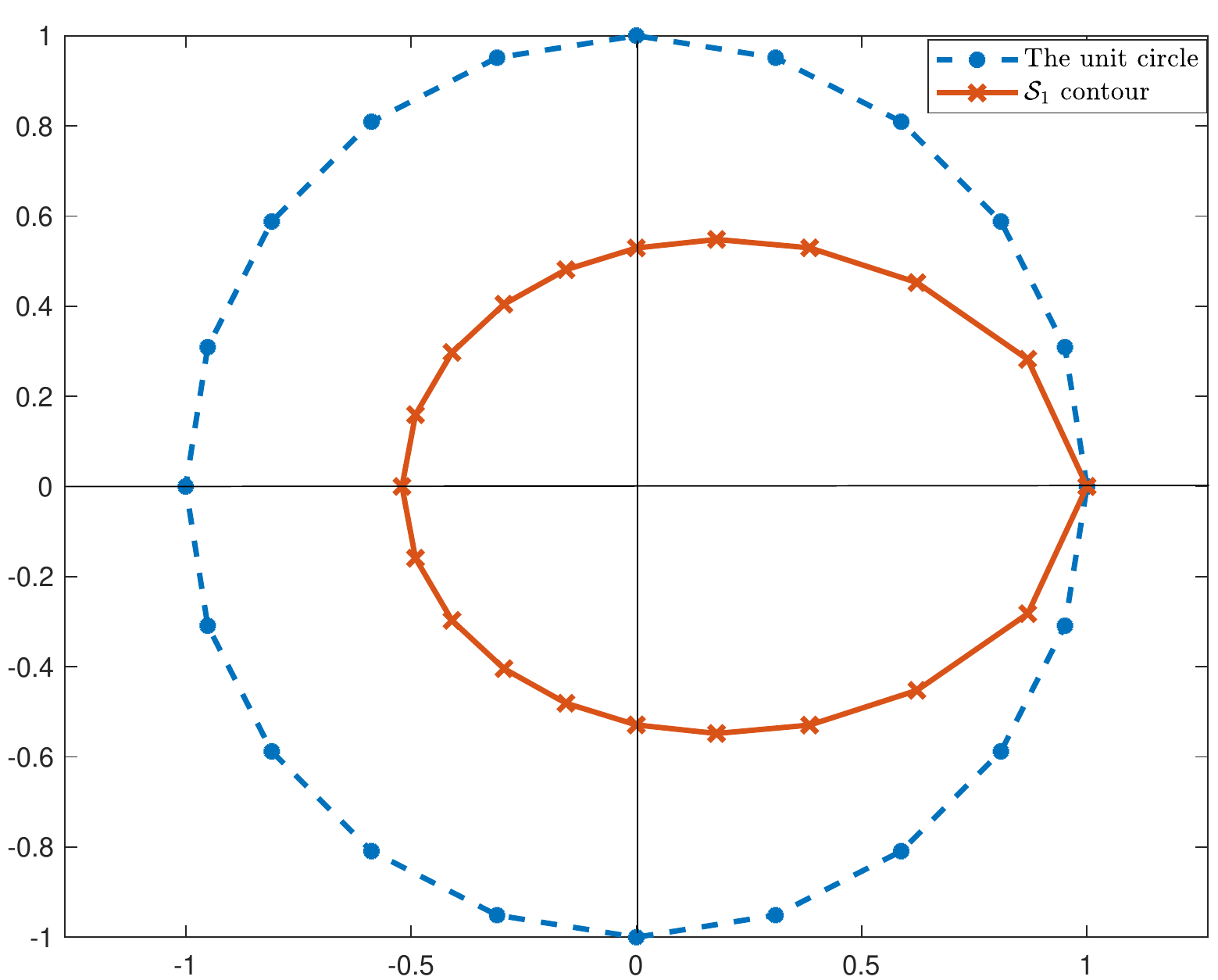}
\caption{The contour $\mathcal{S}_{1}$ for the symmetrical case ($\mathcal{S}_{1}$ coincides with $\mathcal{S}_{2}$) where $p_{0,1}=p_{1,0}$, $p_{-1,0}=p_{0,-1}$, $p_{1,-1}=p_{-1,1}$, and the unit circle.}\label{conto}
\end{figure}
\begin{figure}[ht!]
\centering
\includegraphics[scale=0.6]{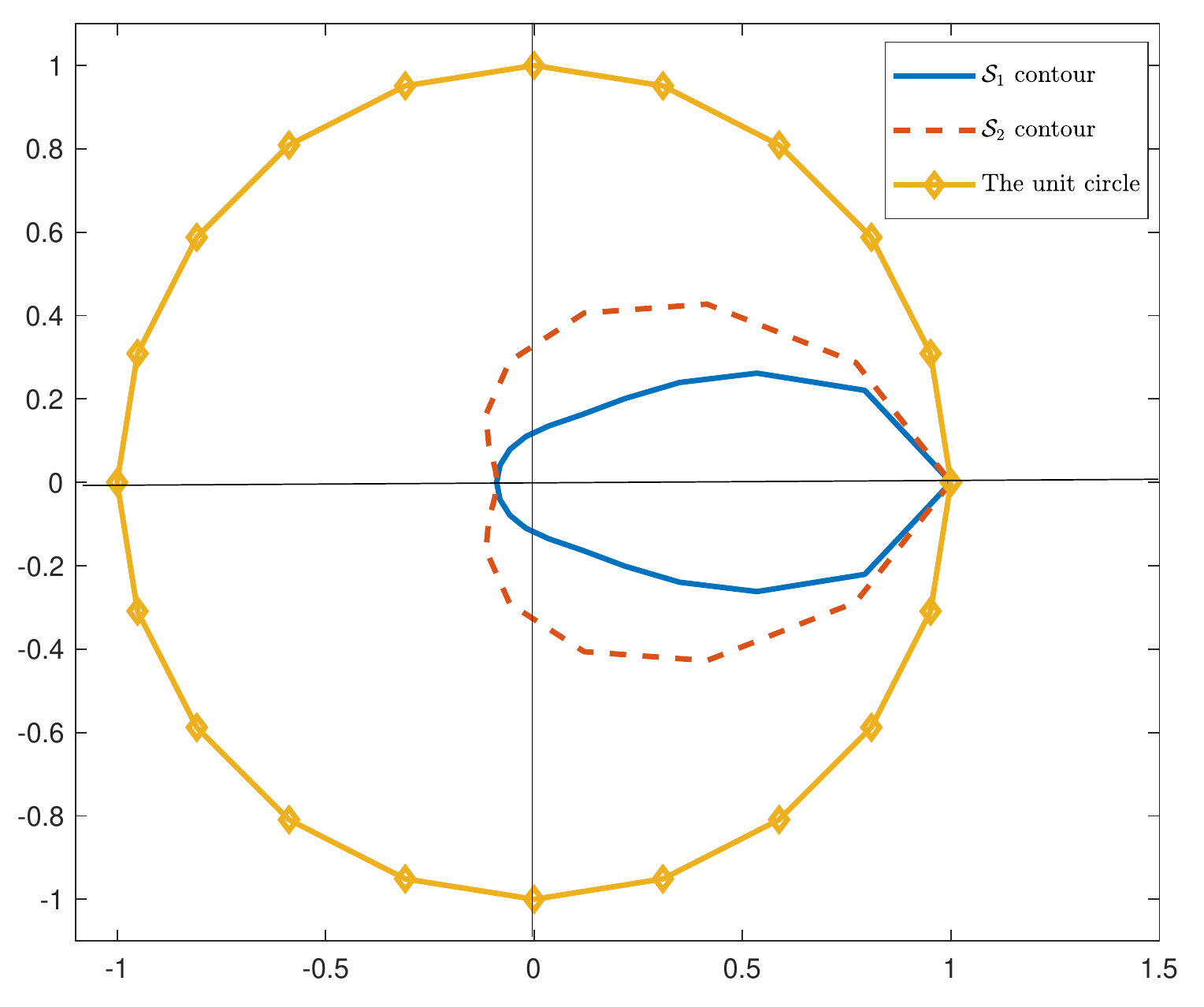}
\caption{The contours $\mathcal{S}_{1}$, $S_{2}$ for the asymmetrical case and the unit circle.}\label{conto1}
\end{figure}
Simple calculations show that for $|g|\leq 1$ satisfying (\ref{fty}),
\begin{equation}
\frac{s}{g}\frac{d}{ds}g(s)=\frac{\mathbb{E}((\xi_{1}^{(3)}-\xi_{2}^{(3)})g^{\xi_{1}^{(3)}+\xi_{2}^{(3)}}s^{\xi_{1}^{(3)}-\xi_{2}^{(3)}})}{\mathbb{E}((2-\xi_{1}^{(3)}-\xi_{2}^{(3)})g^{\xi_{1}^{(3)}+\xi_{2}^{(3)}}s^{\xi_{1}^{(3)}-\xi_{2}^{(3)}})}.
\end{equation}
Thus, if $\mathbb{E}((2-\xi_{1}^{(3)}-\xi_{2}^{(3)})g^{\xi_{1}^{(3)}+\xi_{2}^{(3)}}s^{\xi_{1}^{(3)}-\xi_{2}^{(3)}})\neq0$ for $|g|\leq 1$, $|s|=1$, both zeros of (\ref{fty}) have multiplicity one and each of these zeros is an analytic function of $s$ on the unit circle $|s|=1$.
\begin{theorem}\label{yui}
If $E_{x}<0$, $E_{y}<0$, and $\mathbb{E}((2-\xi_{1}^{(3)}-\xi_{2}^{(3)})g^{\xi_{1}^{(3)}+\xi_{2}^{(3)}}s^{\xi_{1}^{(3)}-\xi_{2}^{(3)}})\neq0$, then $\mathcal{S}_{1}$, $\mathcal{S}_{2}$ are both smooth and analytic contours except possibly at $s=1$. Moreover $x=0\in \mathcal{S}_{1}^{+}$, $y=0\in\mathcal{S}_{2}^{+}$, where $\mathcal{S}_{j}^{+}$ denotes the interior domain bounded by $\mathcal{S}_{j}$, $j=1,2$; see also Figures \ref{conto}, \ref{conto1}.
\end{theorem}
\begin{proof}
The proof follows the lines in Lemma 2.2 in \cite{coh1} and further details are omitted.
\end{proof}

Theorem \ref{yui} implies\footnote{See also theorem 1.1 in \cite{coh2}} that there exists a unique simple contour $\mathcal{L}$ in the $z$-plane with
\begin{displaymath}
z=0\in\mathcal{L}^{+},\,\,z=1\in\mathcal{L},\,\,z=\infty\in\mathcal{L}^{-},
\end{displaymath}
and functions 
\begin{displaymath}
x(z):\mathcal{L}^{+}\cup\mathcal{L}\to\mathcal{S}_{1}^{+}\cup\mathcal{S}_{1},\,\,y(z):\mathcal{L}^{-}\cup\mathcal{L}\to\mathcal{S}_{2}^{+}\cup\mathcal{S}_{2},
\end{displaymath}
such that
\begin{enumerate}
\item $z=0$ is a simple zero of $x(.)$, and $z=\infty$ is a simple zero of $y(.)$, and $0<d:=\lim_{|z|\to\infty}zy(z)<\infty$,
\item $x(z):\mathcal{L}^{+}\to\mathcal{S}_{1}^{+}$ is regular and univalent for $z\in \mathcal{L}^{+}$,
\item $y(z):\mathcal{L}^{-}\to\mathcal{S}_{2}^{+}$ is regular and univalent for $z\in \mathcal{L}^{-}$,
\item $x(z)=f(y(z))$, $z\in\mathcal{L}^{+}$,
\item $x^{+}(z)$, $y^{-}(z)$, $z\in\mathcal{L}$ is a zero pair of the kernel $R(x,y)=0$, with $x^{+}(z)\in\mathcal{S}_{1}$, $y^{+}(z)\in\mathcal{S}_{2}$, where for $z\in\mathcal{L}$, $x^{+}(z)=\lim_{t\to z,t\in\mathcal{L}^{+}}x(t)$, $y^{-}(z)=\lim_{t\to z,t\in\mathcal{L}^{-}}y(t)$.
\end{enumerate}

Thus, following \cite[Section II.3.6]{coh1} for $z\in\mathcal{L}$, there exists a real function $\lambda(z)$ such that $\lambda(1)=0$, and 
\begin{displaymath}
x^{+}(z)=g(e^{i\lambda(z)})e^{i\lambda(z)},\,\,\,\,y^{-}(z)=g(e^{i\lambda(z)})e^{-i\lambda(z)},
\end{displaymath}
and
\begin{equation}
\begin{array}{rl}
x(z)=&ze^{\frac{1}{2i\pi}\int_{\zeta\in\mathcal{L}}log[g(e^{i\lambda(z)})][\frac{\zeta+z}{\zeta-z}-\frac{\zeta+1}{\zeta-1}]\frac{d\zeta}{\zeta}},\,z\in\mathcal{L}^{+},\\
y(z)=&z^{-1}e^{\frac{-1}{2i\pi}\int_{\zeta\in\mathcal{L}}log[g(e^{i\lambda(z)})][\frac{\zeta+z}{\zeta-z}-\frac{\zeta+1}{\zeta-1}]\frac{d\zeta}{\zeta}},\,z\in\mathcal{L}^{-},\\
e^{i\lambda(z)}=&ze^{\frac{1}{2i\pi}\int_{\zeta\in\mathcal{L}}log[g(e^{i\lambda(z)})][\frac{\zeta+z}{\zeta-z}-\frac{\zeta+1}{\zeta-1}]\frac{d\zeta}{\zeta}},\,z\in\mathcal{L}.
\end{array}
\label{ml}
\end{equation}
The last in (\ref{ml}) represents an integral equation for the determination of $\lambda(.)$ and $\mathcal{L}$. In particular, by taking the equivalent integral equation:
\begin{equation}
i\lambda(z)-ln[z]=\frac{1}{2i\pi}\int_{\zeta\in\mathcal{L}}log[g(e^{i\lambda(z)})][\frac{\zeta+z}{\zeta-z}-\frac{\zeta+1}{\zeta-1}]\frac{d\zeta}{\zeta},\,z\in\mathcal{L},
\label{soo}
\end{equation} 
with $\mathcal{L}=\{z:z=\rho(\phi)e^{i\phi},0\leq\phi\leq 2\pi\}$, $\theta(\phi)=\lambda(\rho(\phi)e^{i\phi})$. Separating real and imaginary parts in (\ref{soo}) will lead to two singular integral equations in the two unknowns functions $\rho(.)$, $\theta(.)$. For a numerical treatment of (\ref{soo}), which may be regarded as a generalization of the Theodorsen's integral equation, see \cite[Section IV.2.3]{coh1}.
\subsection{Solution of the functional equation}
Since $(x^{+}(z),y^{-}(z))$, $z\in\mathcal{L}$ is a zero pair of the kernel, it should hold for $z\in\mathcal{L}$
\begin{equation}
A(x^{+}(z),y^{-}(z))g_{0}(x^{+}(z))+B(x^{+}(z),y^{-}(z))h_{0}(y^{-}(z))+C(x^{+}(z),y^{-}(z))=0,
\end{equation}
or equivalently
\begin{equation}
\tilde{g}_{0}(z)=G(z)\tilde{h}_{0}(z)+\tilde{c}(z),\,\,z\in\mathcal{L},
\label{pp}\end{equation}
where $\tilde{g}_{0}(z):=g_{0}(x^{+}(z))$, $\tilde{h}_{0}(z):=h_{0}(y^{-}(z))$ and
\begin{displaymath}
\begin{array}{rl}
G(z):=-\frac{B(x^{+}(z),y^{-}(z))}{A(x^{+}(z),y^{-}(z))},\,\,\,\tilde{c}(z):=-\frac{C(x^{+}(z),y^{-}(z))}{A(x^{+}(z),y^{-}(z))}.
\end{array}
\end{displaymath}

To proceed, we have to ensure that $G(.)$ and $\tilde{c}(.)$ satisfy the Holder condition and $G(.)$ never vanishes. However, the general form of $G(.)$ cannot exclude the possibility both of vanishing, and on taking infinite values at some points of $\mathcal{L}$. More importantly, the poles of $G(.)$ that are located (if any) in the region bounded by $\mathcal{L}$ and the unit circle will be also poles of $\tilde{g}_{0}(.)$. Let $a_{k}$, $k=1,\ldots,m$ the poles of $G(.)$ (i.e., the zeros of $A(x^{+}(z),y^{-}(z))$), with multiplicity $u_{k}$, and let also $b_{s}$, $s=1,\ldots,l$ the zeros of $G(.)$ with multiplicity $p_{s}$. Denote
\begin{displaymath}
\begin{array}{ll}
\widehat{g}_{0}(z):=\prod_{k=1}^{m}(z-a_{k})^{u_{k}}\tilde{g}_{0}(z),&\widehat{h}_{0}(z):=\prod_{s=1}^{l}(z-b_{s})^{p_{s}}\tilde{h}_{0}(z)\\\widehat{G}(z):=\frac{\prod_{s=1}^{l}(z-a_{k})^{u_{k}}}{\prod_{k=1}^{m}(z-b_{c})^{p_{s}}}G(z),&\widehat{c}(z):=\prod_{s=1}^{l}(z-a_{k})^{u_{k}}\tilde{c}(z).
\end{array}
\end{displaymath}
Then, (\ref{pp}) reads for $z\in\mathcal{L}$
\begin{equation}
\widehat{g}_{0}(z)=\widehat{G}(z)\widehat{h}_{0}(z)+\widehat{c}(z),
\label{pqp}
\end{equation}
and is the boundary condition of a non-homogeneous Riemann boundary value problem \cite{ga}. 

If the index $\chi:=index[\widehat{G}(z)]_{\mathcal{L}}\geq 0$, then
\begin{equation}
\begin{array}{rl}
g_{0}(z):=&\prod_{k=1}^{m}(z-a_{k})^{-u_{k}}e^{\Gamma(z)}[\Phi(z)+P_{\chi}(z)],\,z\in\mathcal{L}^{+},\vspace{2mm}\\
h_{0}(z):=&\prod_{s=1}^{l}(z-b_{s})^{-p_{s}}e^{\Gamma(z)}[\Phi(z)+P_{\chi}(z)],\,z\in\mathcal{L}^{-},
\end{array}
\label{solo}
\end{equation}
where 
\begin{displaymath}
\begin{array}{rl}
\Gamma(z):=&\frac{1}{2i\pi}\int_{t\in\mathcal{L}}\log[t^{-\chi}G(t)]\frac{dt}{t-z},\,z\notin\mathcal{L},\vspace{2mm}\\
\Phi(z):=&\frac{1}{2i\pi}\int_{t\in\mathcal{L}}\widehat{c}(t)e^{-\Gamma^{+}(t)}\frac{dt}{t-z},\,z\notin\mathcal{L}.
\end{array}
\end{displaymath}
If $\chi<0$, then 
\begin{equation}
\begin{array}{rl}
g_{0}(z):=&\prod_{k=1}^{m}(z-a_{k})^{-u_{k}}e^{\Gamma(z)}\Phi(z),\,z\in\mathcal{L}^{+},\vspace{2mm}\\
h_{0}(z):=&\prod_{s=1}^{l}(z-b_{s})^{-p_{s}}e^{\Gamma(z)}\Phi(z),\,z\in\mathcal{L}^{-},
\end{array}
\label{solo1}
\end{equation}
but now the following $-\chi-1$ conditions must be satisfied
\begin{displaymath}
\int_{t\in\mathcal{L}}t^{r-1}\widehat{c}(t)e^{-\Gamma^{+}(t)}dt=0,\,r=1,2,...,-\chi-1.
\end{displaymath}
Having obtain $g_{0}(z)$, $h_{0}(z)$ we are able to obtain $g(x,y)$ in \eqref{fun}.

The following steps summarizes the way we can fully determine the stationary distribution: 
\begin{enumerate}
\item The $N_{1}\times N_{2}$ equations for $S_{0}$ involves $(N_{1}+1)\times(N_{2}+1)$ unknowns: $\pi(n_{1},n_{2})$ for $n_{1}=0,1,...,N_{1}$, $n_{2}=0,1,...,N_{2}$. Thus, we further need $N_{1}+N_{2}+1$ equations that involve the unknowns $\pi(N_{1},n_{2})$, $n_{2}=0,1,...,N_{2}-1$, and $\pi(n_{1},N_{2})$, $n_{1}=0,1,...,N_{1}-1$, and $\pi(N_{1},N_{2})$.
\item Note that $g_{0}(x)$, $h_{0}(y)$, are expressed in terms of $A(x,y)$, $B(x,y)$ and $C(x,y)$. The first two are known, and the third one contains $N_{1}+N_{2}+1$ unknown probabilities, i.e., $\pi(N_{1},n_{2})$, $n_{2} =0, 1,...,N_{2} -1$ and $\pi(n_{1},N_{2})$, $n_{1} =0, 1,...,N_{1} -1$, and $\pi(N_{1},N_{2})$. Thus, we need some additional equations. These additional equations are derived as follows at steps 3 and 4.
\item Use (\ref{r2}), (\ref{r3}) to express the unknown probabilities in terms of the values of $g_{0}(x)$, $h_{0}(y)$ at point 0, i.e., $\pi(N_{1},0)=g_{0}(0)$, $\pi(0,N_{2})=h_{0}(0)$ and
\begin{equation}
\begin{array}{rl}
\pi(N_{1},n_{2})=&e_{n_{2}}(0)g_{0}(0)+t_{n_{2}}(0),\,n_{2}=1,...,N_{2}-1,\vspace{2mm}\\
\pi(n_{1},N_{2})=&\tilde{e}_{n_{1}}(0)h_{0}(0)+\tilde{t}_{n_{1}}(0),\,n_{1}=1,...,N_{1}-1.
\end{array}
\end{equation}
This procedure will provide $N_{1}+N_{2}$ equations.
\item The normalization equation yields the last one:
\begin{equation}
\begin{array}{c}
1=\sum_{n_{1}=0}^{N_{1}-1}\sum_{n_{2}=0}^{N_{2}-1}\pi(n_{1},n_{2})+\sum_{n_{1}=0}^{N_{1}-1}h_{n_{1}}(1)+\sum_{n_{2}=0}^{N_{2}-1}g_{n_{2}}(1)+g(1,1)
\end{array}
\label{norq}
\end{equation}
\end{enumerate}

\section{The case where $\Psi(0,0)=0$}\label{fayo}
In the following we focus on the case $\Psi(0,0)=0$, i.e., $p_{-1,-1}=0$, and provide a slightly different analysis for the solution of (\ref{fun}), which is now reduced in terms of a solution of a Riemann-Hilbert boundary value problem. We focus only on the part that is different compared with the previous procedure, and relies on the analysis of the functional equation. We first provide the essential kernel analysis in the following subsection.
\subsection{Kernel analysis}
The kernel $R(x,y)$ is a quadratic polynomial with respect to $x$, $y$. Indeed,
\begin{displaymath}
R(x,y)=\widehat{a}(x)y^2+\widehat{b}(x)y+\widehat{c}(x)=a(y)x^{2}+b(y)x+c(y),
\end{displaymath}
where,
\begin{displaymath}
\begin{array}{rl}
\widehat{a}(x)=&-(xp_{0,1}+x^{2}p_{1,1}+p_{-1,1}),\\
\widehat{b}(x)=&x(1-p_{0,0})-p_{-1,0}-x^{2}p_{1,0},\\
\widehat{c}(x)=&-(x^{2}p_{1,-1}+xp_{0,-1}),\\
a(y)=&-(yp_{1,0}+y^{2}p_{1,1}+p_{1,-1}),\\
b(y)=&y(1-p_{0,0})-p_{0,-1}-y^{2}p_{0,1},\\
c(y)=&-(y^{2}p_{-1,1}+yp_{-1,0}).
\end{array}
\end{displaymath}

In the following we provide some technical lemmas that are necessary for the formulation of a Riemann-Hilbert boundary value problem, the solution of which provides the unknown partial generating functions $g_{0}(x)$, $h_{0}(y)$.
\begin{lemma}\label{lemm1}
For $|y|=1$, $y\neq1$, the kernel equation $R(x,y)=0$ has exactly one root $x=X_{0}(y)$ such that $|X_{0}(y)|<1$. For $\gamma:=p_{1,0}+p_{1,1}+p_{1,-1}-p_{-1,1}-p_{-1,0}<0$, $X_{0}(1)=1$. Similarly, we can prove that $R(x,y)=0$ has exactly one root $y=Y_{0}(x)$, such that $|Y_{0}(x)|\leq1$, for $|x|=1$.
\end{lemma}
\begin{proof}
See Appendix \ref{ap}.
\end{proof}


Next step is to identify the location of the branch points of the two valued function $Y_{\pm}(x)=\frac{-\widehat{b}(x)\pm\sqrt{D_{Y}(x)}}{2\widehat{a}(x)}$ (resp. $X_{\pm}(y)=\frac{-b(y)\pm\sqrt{D_{X}(y)}}{2a(y)}$) defined by $R(x,Y(x))=0$ (resp. $R(X(y),y)=0$), where $D_{Y}(x)=\widehat{b}(x)^{2}-4\widehat{a}(x)\widehat{c}(x)$ (resp. $D_{X}(y)=b(y)^{2}-4a(y)c(y)$). The branch points of $Y_{\pm}(x)$ (resp. $X_{\pm}(y)$) are defined as the roots of $D_{Y}(x)=0$ (resp. $D_{X}(y)=0$).
\begin{lemma}
The algebraic function $Y(x)$, defined by $R(x,Y(x)) = 0$, has four real branch points, say $x_{1},x_{2},x_{3},x_{4}$, such that $x_{1},x_{2}$ lie inside the unit disc, and $x_{3},x_{4}$ lie outside the unit disc. Moreover, $D_{Y}(x)<0$, $x\in(x_{1},x_{2})\cup(x_{3},x_{4})$. 
Similarly, $X(y)$, defined by $R(X(y),y) = 0$, has also four real branch points, $y_{1},y_{2}$ lie inside the unit disc, and $y_{3},y_{4}$ outside the unit disc and $D_{X}(y)<0$, $y\in(y_{1},y_{2})\cup(y_{3},y_{4})$.
\end{lemma}
\begin{proof}
The proof is based on Lemma 2.3.8, pp. 27-28, \cite{fay}, and further details are omitted.
\end{proof}

To ensure the continuity of the function two valued function $Y(x)$ (resp. $X(y)$) we consider the following cut planes: $\doubletilde{C}_{x}=\mathbb{C}_{x}-([x_{1},x_{2}]\cup[x_{3},x_{4}]$, $\doubletilde{C}_{y}=\mathbb{C}_{y}-([y_{1},y_{2}]\cup[y_{3},y_{4}]$, where $\mathbb{C}_{x}$, $\mathbb{C}_{y}$ the complex planes of $x$, $y$, respectively. Let also for $x\in\doubletilde{C}_{x}$
\begin{displaymath}
Y_{0}(x)=\left\{\begin{array}{rl}
Y_{-}(x),&{\text{if }}|Y_{-}(x)|\leq|Y_{+}(x)|,\\
Y_{+}(x),&{\text{if }}|Y_{-}(x)|>|Y_{+}(x)|,
\end{array}\right.\,\,Y_{1}(x)=\left\{\begin{array}{rl}
Y_{+}(x),&{\text{if }}|Y_{-}(x)|\leq|Y_{+}(x)|,\\
Y_{-}(x),&{\text{if }}|Y_{-}(x)|>|Y_{+}(x)|,
\end{array}\right.
\end{displaymath}
i.e., $Y_{0}(.)$ is the zero of $R(x,Y(x))$ with the smallest modulus. Similarly, we can define $X_{0}(y)$, $X_{1}(y)$ in $\doubletilde{C}_{y}$

In $\doubletilde{C}_{x}$ (resp. $\doubletilde{C}_{y}$), denote by $Y_{0}(x)$ (resp. $X_{0}(y)$) the zero of $R(x,Y(x))=0$ (resp. $R(X(y),y)=0$) with the smallest modulus, and $Y_{1}(x)$ (resp. $X_{1}(y)$) the other one. Define also the image contours, $\mathcal{L}=Y_{0}[\overrightarrow{\underleftarrow{x_{1},x_{2}}}]$, $\mathcal{M}=X_{0}[\overrightarrow{\underleftarrow{y_{1},y_{2}}}]$, where $[\overrightarrow{\underleftarrow{u,v}}]$ stands for the contour traversed from $u$ to $v$ along the upper edge of the slit $[u,v]$ and then back to $u$ along the lower edge of the slit. 
The following lemma shows that the mappings $Y(x)$, $X(y)$, for $x\in[x_{1},x_{2}]$, $y\in[y_{1},y_{2}]$ respectively, give rise to the smooth and closed contours $\mathcal{L}$, $\mathcal{M}$ respectively.
\begin{lemma}\label{lem1}\begin{enumerate}\item For $y\in[y_{1},y_{2}]$, the algebraic function $X(y)$ lies on a closed contour $\mathcal{M}$, which is symmetric with respect to the real line and written as a function of $Re(x)$, i.e.,
\begin{displaymath}
\begin{array}{l}
|x|^{2}=m(Re(x)),\,|x|^{2}\leq\frac{c(y_{2})}{a(y_{2})}.
\end{array}
\end{displaymath}
Set $\beta_{0}:=\sqrt{\frac{c(y_{2})}{a(y_{2})}}$, $\beta_{1}=-\sqrt{\frac{c(y_{1})}{a(y_{1})}}$ the extreme right and left point of $\mathcal{M}$, respectively.
\item For $x\in[x_{1},x_{2}]$, the algebraic function $Y(x)$ lies on a closed contour $\mathcal{L}$, which is symmetric with respect to the real line and written as a function of $Re(y)$ as,
\begin{displaymath}
\begin{array}{l}
|y|^{2}=v(Re(y)),\,|y|^{2}\leq\frac{\widehat{c}(x_{2})}{\widehat{a}(x_{2})}.
\end{array}
\end{displaymath}
Set $\eta_{0}:=\sqrt{\frac{\widehat{c}(x_{2})}{\widehat{a}(x_{2})}}$, $\eta_{1}=-\sqrt{\frac{\widehat{c}(x_{1})}{\widehat{a}(x_{1})}}$ the extreme right and left point of $\mathcal{L}$, respectively.
\end{enumerate}
\end{lemma}
\begin{proof}
See Appendix \ref{ap1}. 
\end{proof}
\subsection{Formulation and solution of a Riemann-Hilbert boundary value problem}

For $y\in \mathcal{D}_{y}=\{y\in\mathcal{C}:|y|\leq1,|X_{0}(y)|\leq1\}$,
\begin{equation}
A(X_{0}(y),y)g_{0}(X_{0}(y))+B(X_{0}(y),y)h_{0}(y)+C(X_{0}(y),y)=0.
\label{con}
\end{equation}
For $y\in \mathcal{D}_{y}-[y_{1},y_{2}]$ both $g(X_{0}(y))$, $h_{0}(y)$ are analytic and the right-hand side can be analytically continued up to the slit $[y_{1},y_{2}]$, or equivalently, for $x\in\mathcal{M}$,
\begin{equation}
A(x,Y_{0}(x))g_{0}(x)+B(x,Y_{0}(x))h_{0}(Y_{0}(x))+C(x,Y_{0}(x))=0.
\label{con1}
\end{equation}
Note that $g_{0}(x)$ is holomorphic in $D_{x}=\{x\in\mathbb{C}:|x|<1\}$, and continuous in $\bar{D}_{x}=\{x\in\mathbb{C}:|x|\leq1\}$. However, $g_{0}(x)$ may have poles in $S_{x}=G_{\mathcal{M}}\cap\bar{D}_{x}^{c}$, where $\bar{D}_{x}^{c}=\{x\in\mathbb{C}:|x|>1\}$, and $G_{\mathcal{U}}$ denotes the interior domain bounded by the contour $\mathcal{U}$. These poles (if exist) coincide with the zeros of $A(x,Y_{0}(x))$ in $S_{x}$. 

For $y\in[y_{1},y_{2}]$, let $X_{0}(y)=x\in\mathcal{M}$, and realize that $Y_{0}(X_{0}(y))=y$ \footnote{Without loss of generality we assume that $B(x,Y_{0}(x))\neq0$, $x\in\mathcal{M}$.}. Taking into account the (possible) poles of $g_{0}(x)$ (say, $\xi_{1}$,...,$\xi_{k}$), and noticing that $h_{0}(Y_{0}(x))$ is real for $x\in\mathcal{M}$ we conclude in,
\begin{equation}
Re(iU(x)f(x))=w(x),\,x\in\mathcal{M},
\label{bou}
\end{equation}
where,
\begin{displaymath}
\begin{array}{lcr}
U(x)=\frac{A(x,Y_{0}(x))}{\prod_{i=1}^{k}(x-\xi_{i})B(x,Y_{0}(x))},&f(x)=\prod_{i=1}^{k}(x-\xi_{i})g_{0}(x),&w(x)=Im(\frac{C(x,Y_{0}(x))}{B(x,Y_{0}(x))}).
\end{array}
\end{displaymath}

In order to solve (\ref{bou}), we must first conformally transform it from $\mathcal{M}$ to the unit circle $\mathcal{C}$. Let the mapping, $z=\gamma(x):G_{\mathcal{M}}\to G_{\mathcal{C}}$, and its inverse $x=\gamma_{0}(z):G_{\mathcal{C}}\to G_{\mathcal{M}}$. Then, we have the following problem: Find a function $\tilde{T}(z)=f(\gamma_{0}(z))$ regular for $z\in G_\mathcal{C}$, and continuous for $z\in\mathcal{C}\cup G_\mathcal{C}$ such that, 
\begin{equation}
Re(iU(\gamma_{0}(z))\tilde{T}(z))=w(\gamma_{0}(z)),\,z\in\mathcal{C}.
\end{equation}

To obtain the conformal mappings, we need to represent $\mathcal{M}$ in polar coordinates, i.e., $\mathcal{M}=\{x:x=\rho(\phi)\exp(i\phi),\phi\in[0,2\pi]\}.$ This procedure is described in detail in \cite{coh1}. We briefly summarized the basic steps: Since $0\in G_{\mathcal{M}}$, for each $x\in\mathcal{M}$, a relation between its absolute value and its real part is given by $|x|^{2}=m(Re(x))$ (see Lemma \ref{lem1}). Given the angle $\phi$ of some point on $\mathcal{M}$, the real part of this point, say $\delta(\phi)$, is the solution of $\delta-\cos(\phi)\sqrt{m(\delta)}$, $\phi\in[0,2\pi].$ Since $\mathcal{M}$ is a smooth, egg-shaped contour, the solution is unique. Clearly, $\rho(\phi)=\frac{\delta(\phi)}{\cos(\phi)}$, and the parametrization of $\mathcal{M}$ in polar coordinates is fully specified. Then, the mapping from $z\in G_{\mathcal{C}}$ to $x\in G_{\mathcal{M}}$, where $z = e^{i\phi}$ and $x= \rho(\tilde{\psi}(\phi))e^{i\tilde{\psi}(\phi)}$, satisfying $\gamma_{0}(0)=0$ and $\gamma_{0}(z)=\overline{\gamma_{0}(\overline{z})}$ is uniquely determined by (see \cite{coh1}, Section I.4.4),
\begin{equation}
\begin{array}{rl}
\gamma_{0}(z)=&z\exp[\frac{1}{2\pi}\int_{0}^{2\pi}\log\{\rho(\tilde{\psi}(\omega))\}\frac{e^{i\omega}+z}{e^{i\omega}-z}d\omega],\,|z|<1,\\
\tilde{\psi}(\phi)=&\phi-\int_{0}^{2\pi}\log\{\rho(\psi(\omega))\}\cot(\frac{\omega-\phi}{2})d\omega,\,0\leq\phi\leq 2\pi,
\end{array}
\label{zx}
\end{equation}
i.e., $\tilde{\psi}(.)$ is uniquely determined as the solution of a Theodorsen integral equation with $\tilde{\psi}(\phi)=2\pi-\psi(2\pi-\phi)$. Due to the correspondence-boundaries theorem, $\gamma_{0}(z)$ is continuous in $\mathcal{C}\cup G_{\mathcal{C}}$. 

The solution of the boundary value problem depends on its index $\chi=\frac{-1}{\pi}[arg\{U(x)\}]_{x\in \mathcal{M}}$, where $[arg\{U(x)\}]_{x\in \mathcal{M}}$, denotes the variation of the argument of the function $U(x)$ as $x$ moves along $\mathcal{M}$ in the positive direction, provided that $U(x)\neq0$, $x\in\mathcal{M}$.


If $\chi\leq 0$ our problem has at most one linearly independent solution. The solution of the problem defined in (\ref{bou}) is given for $z\in \mathcal{C}_{x}^{+}$ by,
\begin{equation}
\begin{array}{l}
g_{0}(\gamma_{0}(z))=\prod_{i=1}^{k}(\gamma_{0}(z)-\xi_{i})^{-1}e^{i\sigma(z)}z^{\chi}[iK+\frac{1}{2\pi i}\int_{|t|=1}e^{\omega_{1}(t)}\delta(t)\frac{t+z}{t-z}\frac{dt}{t}],
\end{array}
\label{sool1}
\end{equation}
where $K$ is a constant to be determined, $\omega_{1}(z)=Im(\sigma(z))$, $\delta(z)=\frac{w(\gamma_{0}(z))}{|U(\gamma_{0}(z))|}$, $U(\gamma_{0}(z))=b_{1}(z)+ia_{1}(z)$ and
\begin{displaymath}
\begin{array}{rl}
\sigma(z)=&\frac{1}{2\pi i}\int_{|t|=1}(\arctan\frac{b_{1}(t)}{a_{1}(t)}-\chi\arg t)\frac{t+z}{t-z}\frac{dt}{t}.
\end{array}
\end{displaymath}

Note that $g_{0}(x)=g_{0}(\gamma_{0}(\gamma(x)))$. When $\chi=0$, $K$ can be determined from the solution to $g_{0}(0)$. If $\chi<0$, then $K=0$ and a solution exists if \cite{ga}
\begin{displaymath}
\frac{1}{2\pi i}\int_{|t|=1}e^{\omega_{1}(t)}\delta(t)t^{-k-1}dt=0,
\end{displaymath}
for $k=0,1,...,-\chi-1.$ 

The following steps summarizes the way we can fully determine the stationary distribution: 
\begin{enumerate}
\item The $N_{1}\times N_{2}$ equations for $S_{0}$ involves $(N_{1}+1)\times(N_{2}+1)$ unknowns: $\pi(n_{1},n_{2})$ for $n_{1}=0,1,...,N_{1}$, $n_{2}=0,1,...,N_{2}$, excluding $\pi(N_{1},N_{2})$. Thus, we further need $N_{1}+N_{2}$ equations that involve the unknowns $\pi(N_{1},n_{2})$, $n_{2}=0,1,...,N_{2}-1$, and $\pi(n_{1},N_{2})$, $n_{1}=0,1,...,N_{1}-1$.
\item $g_{0}(x)$, $h_{0}(y)$, are expressed in terms of $A(x,y)$, $B(x,y)$ and $C(x,y)$, where the third one contains $N_{1}+N_{2}+1$ unknown probabilities, i.e., $\pi(N_{1},n_{2})$, $n_{2} =0, 1,...,N_{2} -1$ and $\pi(n_{1},N_{2})$, $n_{1} =0, 1,...,N_{1} -1$, and $\pi(N_{1},N_{2})$. These additional equations are now derived as follows in steps 3 and 4.
\item Use (\ref{r2}), (\ref{r3}) to express the unknown probabilities in terms of the derivatives of $g_{0}(x)$, $h_{0}(y)$ at point 0, i.e.,
\begin{equation}
\begin{array}{rl}
n_{2}!\pi(N_{1},n_{2})=&\frac{d^{n_{2}}}{dx^{n_{2}}}[e_{n_{2}}(x)g_{0}(x)+t_{n_{2}}(x)]|_{x=0},\,n_{2}=1,...,N_{2},\vspace{2mm}\\
n_{1}!\pi(n_{1},N_{2})=&\frac{d^{n_{1}}}{dy^{n_{1}}}[\tilde{e}_{n_{1}}(y)h_{0}(y)+\tilde{t}_{n_{1}}(y)]|_{y=0},\,n_{1}=1,...,N_{1},\end{array}\label{fc}
\end{equation}
where now, for $n_{2}=1,2,...,N_{2}$
\begin{displaymath}
\begin{array}{rl}
e_{n_{2}}(x)=&\frac{f_{2}(N_{1},n_{2}-1,x)e_{n_{2}-1}(x)-f_{1}(N_{1},n_{2}-2,x)e_{n_{2}-2}(x)}{p_{0,-1}(N_{1},n_{2})+xp_{1,-1}(N_{1},n_{2})},\vspace{2mm}\\
t_{n_{2}}(x)=&\frac{f_{2}(N_{1},n_{2}-1,x)t_{n_{2}-1}(x)-f_{1}(N_{1},n_{2}-2,x)t_{n_{2}-2}(x)-b_{n_{2}-1}(x)}{p_{0,-1}(N_{1},n_{2})+xp_{1,-1}(N_{1},n_{2})},
\end{array}
\end{displaymath}
and for $n_{1}=1,2,...,N_{1}$,
\begin{displaymath}
\begin{array}{rl}
\tilde{e}_{n_{1}}(y)=&\frac{\tilde{f}_{2}(n_{1}-1,N_{2},y)\tilde{e}_{n_{1}-1}(y)-\tilde{f}_{1}(n_{1}-2,N_{2},y)\tilde{e}_{n_{1}-2}(y)}{p_{-1,0}(n_{1},N_{2})+yp_{-1,1}(n_{1},N_{2})},\vspace{2mm}\\
\tilde{t}_{n_{1}}(y)=&\frac{\tilde{f}_{2}(n_{1}-1,N_{2},y)\tilde{t}_{n_{1}-1}(y)-\tilde{f}_{1}(n_{1}-2,N_{2},y)\tilde{t}_{n_{1}-2}(y)-u_{n_{1}-1}(y)}{p_{-1,0}(n_{1},N_{2})+yp_{-1,1}(n_{1},N_{2})}.
\end{array}
\end{displaymath}
This procedure will provide the $N_{1}+N_{2}$ equations referred at step $1$. 
\item The last additional equation for the determination of the last unknown $\pi(N_{1},N_{2})$ is done by the use of the normalization equation:
\begin{equation}
\begin{array}{c}
1=\sum_{n_{1}=0}^{N_{1}-1}\sum_{n_{2}=0}^{N_{2}-1}\pi(n_{1},n_{2})+\sum_{n_{1}=0}^{N_{1}-1}h_{n_{1}}(1)+\sum_{n_{2}=0}^{N_{2}-1}g_{n_{2}}(1)+g(1,1)
\end{array}
\label{nor}
\end{equation}
\end{enumerate}


\section{Application: An adaptive ALOHA-type random access network}\label{appl}
In the following we present an interesting application of PH-NNRWQP in the modelling of queue-aware multiple access systems. The analysis of such a system can be done following the lines of Section \ref{fayo}. 

Consider an ALOHA-type wireless network with two users communicating with a common destination node; see Figure \ref{fig1w}. Each user is equipped with an infinite capacity buffer for storing arriving and backlogged packets. The packet arrival processes are assumed to be independent from user to user and the channel is slotted in time, with a slot period to be equal the packet length. 

Let $Q_{k,m}$, $k=1,2,$ be the number of stored packets at the buffer of user $k$, at the beginning of the $m$th slot. Then $\mathbf{Q}_{m}=\{(Q_{1,m},Q_{2,m}),m=0,1,...\}$ is a two-dimensional discrete time Markov chain with state space $S=\{\underline{n}=(n_{1},n_{2});n_{k}\geq0,k=1,2\}$.
\paragraph{Transmission control:} At the beginning of each slot, given that the sate of the network is $\underline{n}$, user node $k$, $k=1,2$ transmits a packet to the destination node with probability $a_{k}(\underline{n})$ (with prob. $\bar{a}_{k}(\underline{n})$ remains silent). If both user nodes transmit at the same slot there is a collision, and both packets have to be retransmitted in a later slot. Packet arrivals are assumed i.i.d. random variables from slot to slot, both depended on the state of the network at the beginning of a slot. 

Let $A_{k,m}(\underline{n})$ the number of packets that arrive at $(m,m+1]$ given that at the beginning of the $m$th slot the state of the network is $\underline{n}$. We assume Bernoulli arrivals with the average number of arrivals being $\mathbb{E}(A_{k,m}(\underline{n}))=\lambda_{k}(\underline{n})<\infty$ packets per slot.
\begin{figure}[htp]
\centering
\includegraphics[scale=0.55]{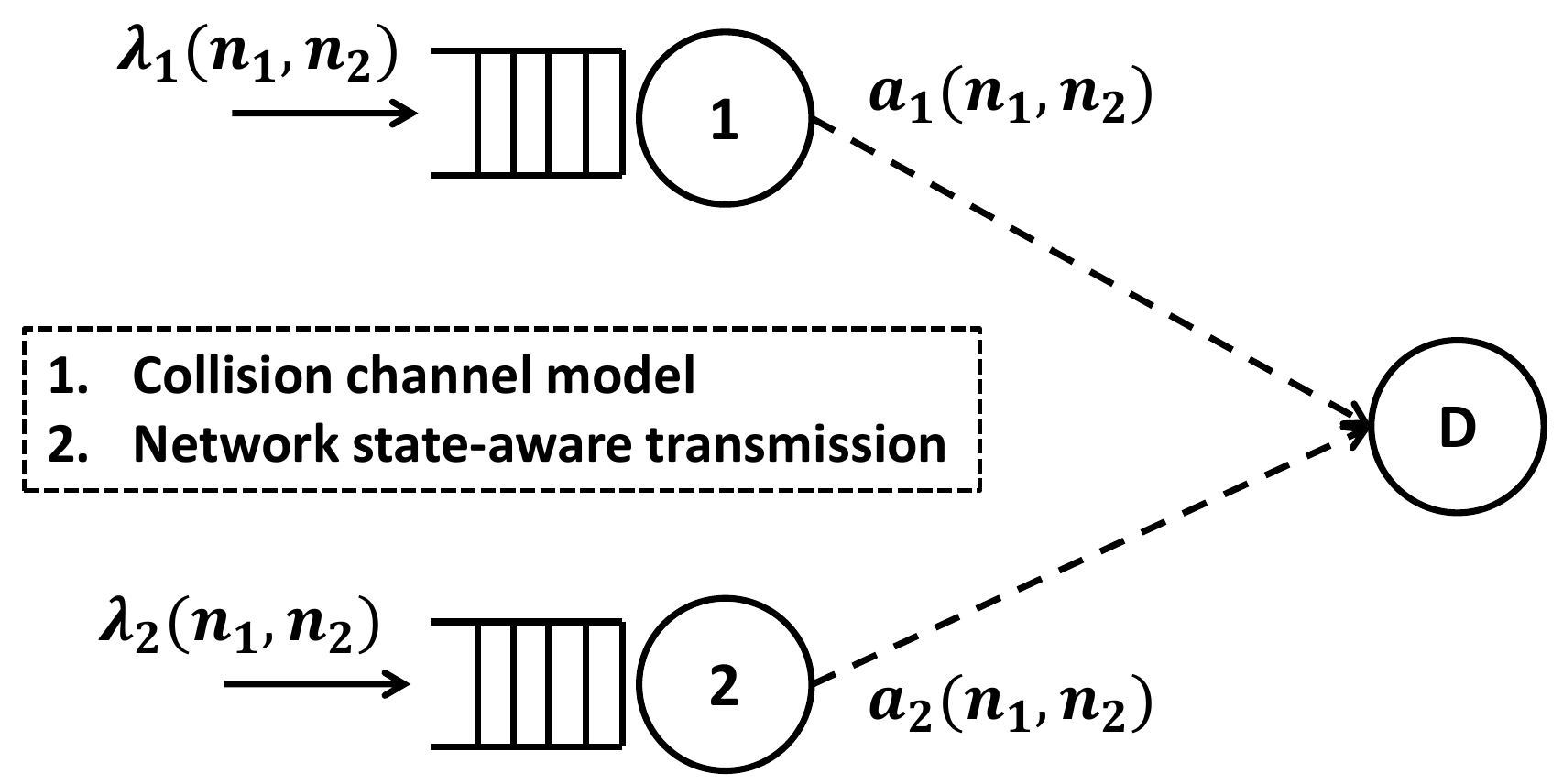}
\caption{The adaptive ALOHA-type network.}
\label{fig1w}
\end{figure}

We consider a \textit{limited-state dependent} queue-based transmission protocol. In particular, we assume that there exist two positive constants, say $N_{1}$, $N_{2}$, such that they split the state space $S$ in four non-intersecting subsets
\begin{displaymath}
\begin{array}{rl}
S_{0}=\{(n_{1},n_{2});n_{1}<N_{1},n_{2}<N_{2}\},&S_{1}=\{(n_{1},n_{2});n_{1}\geq N_{1},n_{2}<N_{2}\},\\
S_{2}=\{(n_{1},n_{2});n_{1}<N_{1},n_{2}\geq N_{2}\},&S_{3}=\{(n_{1},n_{2});n_{1}\geq N_{1},n_{2}\geq N_{2}\},
\end{array}
\end{displaymath}
and assume that for $k=1,2,$
\begin{displaymath}
a_{k}(\underline{n})=\begin{cases}
\begin{array}{rl}
a_{k}(N_{1},n_{2}),&if\,\underline{n}\in S_{1},\\
a_{k}(n_{1},N_{2}),&if\,\underline{n}\in S_{2},\\
a_{k}(N_{1},N_{2}),&if\,\underline{n}\in S_{3},\\
\end{array}\end{cases}\,\,\lambda_{k}(\underline{n})=\begin{cases}
\begin{array}{rl}
\lambda_{k}(N_{1},n_{2}),&if\,\underline{n}\in S_{1},\\
\lambda_{k}(n_{1},N_{2}),&if\,\underline{n}\in S_{2},\\
\lambda_{k}(N_{1},N_{2}),&if\,\underline{n}\in S_{3}.\\
\end{array}\end{cases}
\end{displaymath}
 The one step transition probabilities from $\underline{n}=(n_{1},n_{2})$ to $(n_{1}+i,n_{2},+j)$, say $p_{i,j}(\underline{n})$, where, $\underline{n}\in S$, $i,j=-1,0,1$, are given by:
\begin{displaymath}
\begin{array}{rl}
p_{1,0}(\underline{n})=&(\bar{a}_{1}(\underline{n})\bar{a}_{2}(\underline{n})+a_{1}(\underline{n})a_{2}(\underline{n}))d_{1,0}(\underline{n})+\bar{a}_{1}(\underline{n})a_{2}(\underline{n})d_{1,1}(\underline{n}),\\
p_{0,1}(\underline{n})=&(\bar{a}_{1}(\underline{n})\bar{a}_{2}(\underline{n})+a_{1}(\underline{n})a_{2}(\underline{n}))d_{0,1}(\underline{n})+\bar{a}_{2}(\underline{n})a_{1}(\underline{n})d_{1,1}(\underline{n}),\\
p_{1,1}(\underline{n})=&(\bar{a}_{1}(\underline{n})\bar{a}_{2}(\underline{n})+a_{1}(\underline{n})a_{2}(\underline{n}))d_{1,1}(\underline{n}),\\
p_{-1,1}(\underline{n})=&a_{1}(\underline{n})\bar{a}_{2}(\underline{n})d_{0,1}(\underline{n}),\\
p_{1,-1}(\underline{n})=&a_{2}(\underline{n})\bar{a}_{1}(\underline{n})d_{1,0}(\underline{n}),\\
p_{-1,0}(\underline{n})=&a_{1}(\underline{n})\bar{a}_{2}(\underline{n})d_{0,0}(\underline{n}),\\
p_{0,-1}(\underline{n})=&a_{2}(\underline{n})\bar{a}_{1}(\underline{n})d_{0,0}(\underline{n}),\\
p_{0,0}(\underline{n})=&(\bar{a}_{1}(\underline{n})\bar{a}_{2}(\underline{n})+a_{1}(\underline{n})a_{2}(\underline{n}))d_{0,0}(\underline{n})+\bar{a}_{1}(\underline{n})a_{2}(\underline{n})d_{0,1}(\underline{n})\\&+\bar{a}_{2}(\underline{n})a_{1}(\underline{n})d_{1,0}(\underline{n}),
\end{array}
\end{displaymath}
where 
\begin{displaymath}
d_{i,j}(\underline{n})=\left\{\begin{array}{rl}
\lambda_{1}(\underline{n})\bar{\lambda}_{2}(\underline{n}),&i=1,j=0,\\
\lambda_{2}(\underline{n})\bar{\lambda}_{1}(\underline{n}),&i=0,j=1,\\
\lambda_{1}(\underline{n})\lambda_{2}(\underline{n}),&i=1,j=1,\\
\bar{\lambda}_{1}(\underline{n})\bar{\lambda}_{2}(\underline{n}),&i=0,j=0.
\end{array}\right.
\end{displaymath}
and $\bar{\lambda}_{k}(\underline{n})=1-\lambda_{k}(\underline{n})$, $k=1,2$, $a_{1}(0,n_{2})=0=a_{2}(n_{1},0)$.

For $(n_{1},n_{2})\in S_{3}$, $a_{k}(n_{1},n_{2})=a_{k}(N_{1},N_{2}):=a_{k}$, $\lambda_{k}(n_{1},n_{2})=\lambda_{k}(N_{1},N_{2}):=\lambda_{k}$, $k=1,2,$ and equation (\ref{fun}) takes the following form
\begin{equation}
\begin{array}{rl}
R(x,y)=&xy-D(x,y)[xy+a_{1}\bar{a}_{2}y(1-x)+\bar{a}_{1}a_{2}x(1-y)],\\
D(x,y)=&(\bar{\lambda}_{1}+\lambda_{1}x)(\bar{\lambda}_{2}+\lambda_{2}y).
\end{array}
\end{equation}
The rest of the analysis follows the lines of Section \ref{fayo}.
\subsection{Ergodicity conditions}\label{sta}
Note that our model is described by a two-dimensional Markov with limited state dependency, or equivalently with partial spatial homogeneity. The ergodicity conditions for the model at hand reads as follows.

For $Q_{1,m}>N_{1}$ (resp. $Q_{2,m}>N_{2}$) the component $Q_{2,m}$ (resp. $Q_{1,m}$) evolves as a one-dimensional RW. Denote its corresponding stationary distribution by $\psi:=(\psi_{0},\psi_{1},...)$ (resp. $\varphi:=(\varphi_{0},\varphi_{1},...)$); see Appendix \ref{apc} for details on the corresponding induced Markov chains. Consider now the mean drifts
\begin{displaymath}
\begin{array}{rl}
\gamma_{n_{2}}:=&\mathbb{E}(Q_{1,m+1}-Q_{1,m}|\mathbf{Q}_{m}=(n_{1},n_{2}))\\&=\lambda_{1}(N_{1},n_{2})-a_{1}(N_{1},n_{2})\bar{a}_{2}(N_{1},n_{2}),\,\forall n_{1}\geq N_{1},\vspace{2mm}\\
\delta_{n_{1}}:=&\mathbb{E}(Q_{2,m+1}-Q_{2,m}|\mathbf{Q}_{m}=(n_{1},n_{2}))=\\&=\lambda_{2}(n_{1},N_{2})-a_{2}(n_{1},N_{2})\bar{a}_{1}(n_{1},N_{2}),\,\forall n_{2}\geq N_{2}.
\end{array}
\end{displaymath}
Since $a_{k}(\underline{n}):=a_{k}$, $\lambda_{k}(\underline{n})=\lambda_{k}$, $k=1,2,$ for $\underline{n}\in S_{3}=\{(n_{1},n_{2}):n_{1}\geq N_{1},n_{2}\geq N_{2}\}$,
\begin{displaymath}
\begin{array}{rl}
\gamma_{n_{2}}:=&\gamma=\lambda_{1}-a_{1}\bar{a}_{2},\,n_{2}\geq N_{2},\\
\delta_{n_{1}}:=&\delta=\lambda_{2}-a_{2}\bar{a}_{1},\,n_{1}\geq N_{1}.
\end{array}
\end{displaymath}
Then, the following theorem provides necessary and sufficient conditions for ergodicity \cite{fay2}. For a similar approach, see \cite{za}\footnote{Note also that for $N_{1}=N_{2}=1$, Theorem \ref{ergth} coincides with the well known ergodicity result presented in Theorem 3.3.1 in \cite{fay}}.
\begin{theorem}\label{ergth}
\begin{enumerate}
\item If $\lambda_{1}<a_{1}\bar{a}_{2}$, $\lambda_{2}<a_{2}\bar{a}_{1}$, $\mathbf{Q}_{m}$ is 
\begin{enumerate}
\item ergodic if
\begin{equation}
\begin{array}{rl}
\lambda_{1}(1-\sum_{k=0}^{N_{2}-1}\psi_{k})<&a_{1}\bar{a}_{2}(1-\sum_{k=0}^{N_{2}-1}\psi_{k})-\sum_{k=0}^{N_{2}-1}\gamma_{k}\psi_{k},\,and\vspace{2mm}\\
\lambda_{2}(1-\sum_{k=0}^{N_{1}-1}\varphi_{k})<&a_{2}\bar{a}_{1}(1-\sum_{k=0}^{N_{1}-1}\varphi_{k})-\sum_{k=0}^{N_{1}-1}\delta_{k}\varphi_{k}.
\end{array}
\end{equation}
\item transient if
\begin{equation}
\begin{array}{rl}
\lambda_{1}(1-\sum_{k=0}^{N_{2}-1}\psi_{k})>&a_{1}\bar{a}_{2}(1-\sum_{k=0}^{N_{2}-1}\psi_{k})-\sum_{k=0}^{N_{2}-1}\gamma_{k}\psi_{k},\,or\vspace{2mm}\\
\lambda_{2}(1-\sum_{k=0}^{N_{1}-1}\varphi_{k})>&a_{2}\bar{a}_{1}(1-\sum_{k=0}^{N_{1}-1}\varphi_{k})-\sum_{k=0}^{N_{1}-1}\delta_{k}\varphi_{k}.
\end{array}
\end{equation}
\end{enumerate}
\item If $\lambda_{1}\geq a_{1}\bar{a}_{2}$, $\lambda_{2}<a_{2}\bar{a}_{1}$, $\mathbf{Q}_{m}$ is
\begin{enumerate}
\item ergodic if
\begin{equation*}
\begin{array}{rl}
\lambda_{1}(1-\sum_{k=0}^{N_{2}-1}\psi_{k})<&a_{1}\bar{a}_{2}(1-\sum_{k=0}^{N_{2}-1}\psi_{k})-\sum_{k=0}^{N_{2}-1}\gamma_{k}\psi_{k}.
\end{array}
\end{equation*}
\item transient if 
\begin{equation*}
\begin{array}{rl}
\lambda_{1}(1-\sum_{k=0}^{N_{2}-1}\psi_{k})>&a_{1}\bar{a}_{2}(1-\sum_{k=0}^{N_{2}-1}\psi_{k})-\sum_{k=0}^{N_{2}-1}\gamma_{k}\psi_{k},
\end{array}
\end{equation*}
or when $\lambda_{1}> a_{1}\bar{a}_{2}$ and $\lambda_{1}(1-\sum_{k=0}^{N_{2}-1}\psi_{k})=a_{1}\bar{a}_{2}(1-\sum_{k=0}^{N_{2}-1}\psi_{k})-\sum_{k=0}^{N_{2}-1}\gamma_{k}\psi_{k}$.
\end{enumerate}

\item If $\lambda_{1}< a_{1}\bar{a}_{2}$, $\lambda_{2}\geq a_{2}\bar{a}_{1}$, $\mathbf{Q}_{m}$ is
\begin{enumerate}
\item  ergodic if 
\begin{equation*}
\begin{array}{rl}
\lambda_{2}(1-\sum_{k=0}^{N_{1}-1}\varphi_{k})<&a_{2}\bar{a}_{1}(1-\sum_{k=0}^{N_{1}-1}\varphi_{k})-\sum_{k=0}^{N_{1}-1}\delta_{k}\varphi_{k}.
\end{array}
\end{equation*}
\item transient if
\begin{equation*}
\begin{array}{rl}
\lambda_{2}(1-\sum_{k=0}^{N_{1}-1}\varphi_{k})>&a_{2}\bar{a}_{1}(1-\sum_{k=0}^{N_{1}-1}\varphi_{k})-\sum_{k=0}^{N_{1}-1}\delta_{k}\varphi_{k},
\end{array}
\end{equation*}
or when $\lambda_{2}>a_{2}\bar{a}_{1}$ and $\lambda_{2}(1-\sum_{k=0}^{N_{1}-1}\varphi_{k})=a_{2}\bar{a}_{1}(1-\sum_{k=0}^{N_{1}-1}\varphi_{k})-\sum_{k=0}^{N_{1}-1}\delta_{k}\varphi_{k}$.
\end{enumerate}
\item If $\lambda_{1}\geq a_{1}\bar{a}_{2}$, $\lambda_{2}\geq a_{2}\bar{a}_{1}$, $\mathbf{Q}_{m}$ is transient.
\end{enumerate}
\end{theorem}
\begin{proof}
The proof is based on the construction of quadratic Lyapunov functions following the lines in \cite[Theorem 3.1, p. 178]{fay2}. 
\end{proof}

\section{Numerical example}\label{num}
For the numerical illustration, we focus on the application model developed in Section \ref{appl}, by considering an adaptive slotted Aloha network of two users with collisions, which is described by a PH-NNRWQP with no transitions to the South-West; see Section \ref{fayo}. 

\paragraph{Queueing analysis:} As we have seen so far, in order to provide the exact information about the stationary joint queue length distribution at users' queue we have firstly to solve a system of $(N_{1}+1)\times(N_{2}+1)$ linear equations.
\begin{enumerate}
\item  $N_{1}\times N_{2}$ of them refer to the states in region $S_{0}$.
\item $N_{1}+N_{2}$ refer to the equations that correspond to the derivatives
\begin{displaymath}
\begin{array}{rl}
n_{2}!\pi_{N_{1},n_{2}}=\frac{d^{n_{2}}}{dx^{n_{2}}}[e_{n_{2}}(x)g_{0}(x)+t_{n_{2}}(x)]|_{x=0},\,n_{2}=1,...,N_{2},\vspace{2mm}\\
n_{1}!\pi_{n_{1},N_{2}}=\frac{d^{n_{1}}}{dy^{n_{1}}}[\tilde{e}_{n_{1}}(y)h_{0}(y)+\tilde{t}_{n_{1}}(y)|_{y=0},\,n_{1}=1,...,N_{1}.
\end{array}
\end{displaymath}
\item The normalizing equation (\ref{nor}). Moreover, note that each coefficient in the last $N_{1}+N_{2}+1$ equations requires the evaluation of complex integrals of type (\ref{sool1}). In order to numerically evaluate them, we have firstly to construct the conformal mappings. Note that in most of the cases we are not be able to obtain them explicitly. However, an efficient numerical approach was developed in \cite{coh1}, Sec. IV.1.1. Alternatively, since contours are close to ellipses, we can use the nearly circular approximation, \cite{neh}. Function $U(x)$ on which (\ref{bou}) is based, involves determinants of matrices whose elements are polynomials.
\item Solve the functional equation (\ref{fun}).
\end{enumerate} 

In the following we provide a simple numerical example to illustrate our theoretical findings. For ease of computations we focus on the symmetrical system: Set $N_{1}=N_{2}=2$, and for let $||n||=n_{1}+n_{2}$, $a_{k}(\underline{n}):=a(\underline{n})=a\frac{n_{k}}{||n||}$, $\lambda_{k}(\underline{n}):=\lambda(\underline{n})=\lambda 2^{-||n||}$, with $a(\underline{n})=a$, and $\lambda(\underline{n})=\lambda$, for $\underline{n}\in S_{3}$.

In particular, in Figure \ref{fig2}, the total expected number of buffered packets is presented as a function of $\lambda$, $a$. Definitely, by increasing $a$, the delay in queue can be handled as long as $\lambda$ remains in small values. However, we can see there is no significant benefit. This is because by increasing $a$, we also increase the possibility of a collision, which result in unsuccessful transmission. However, by increasing $\lambda$, we observe the increase on the total expected number of buffered packets, as expected.

\begin{figure}
\centering
\includegraphics[scale=0.65]{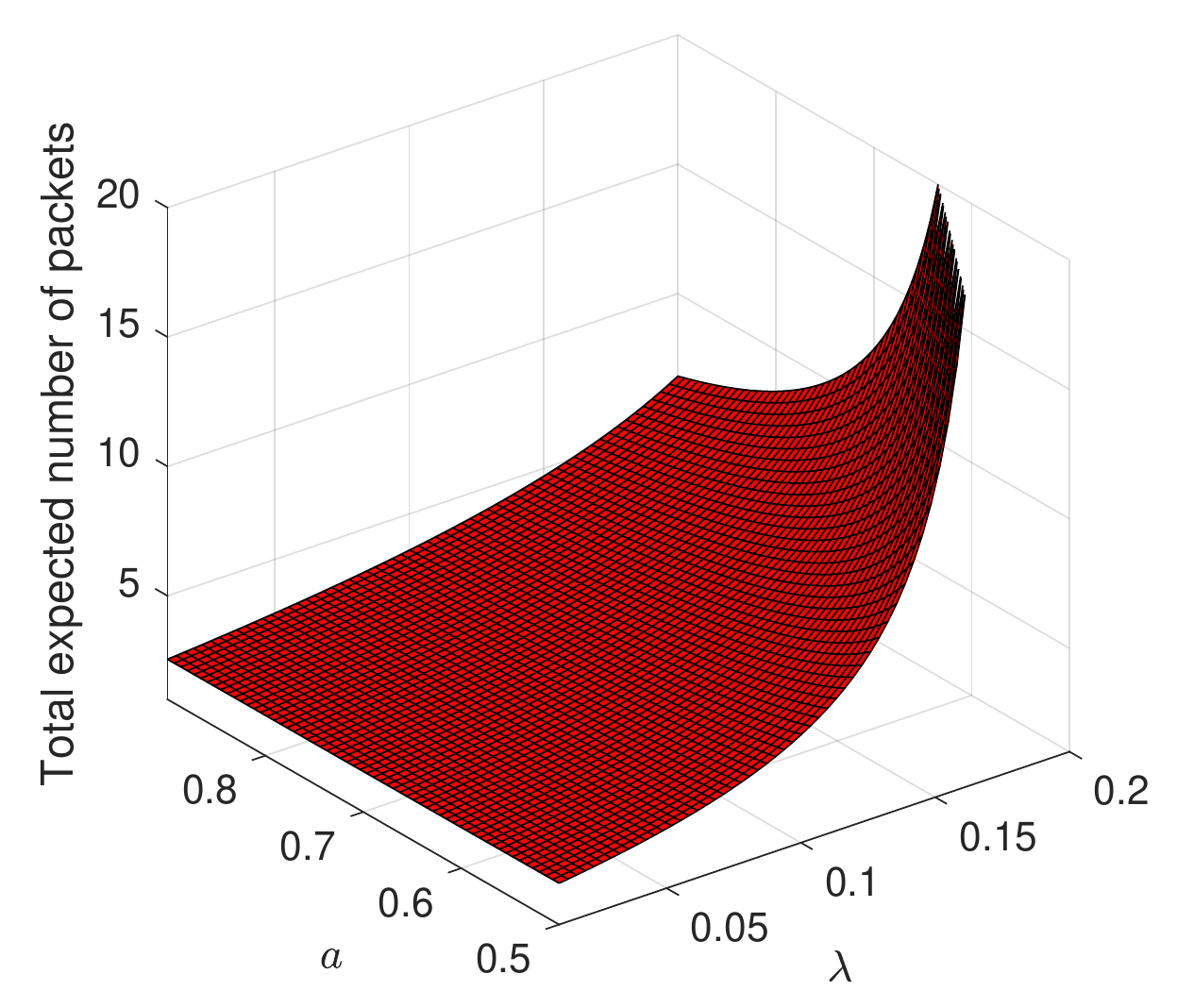}
\caption{Effect of $\lambda$, $a$ on the  total expected number of buffered packets.}\label{fig2}
\end{figure}

\paragraph{Stability condition:} Set $N_{1}=N_{2}=2$, $\lambda_{2}(\underline{n})=0.5^{n_{2}}$, $n_{1}=0,1$, $n_{2}=0,1,2$, and $a_{1}(\underline{n})=0.8\frac{n_{1}}{n_{1}+n_{2}}$, $a_{2}(\underline{n})=0.6\frac{n_{2}}{n_{1}+n_{2}}$. Let also
\begin{equation*}
\begin{array}{rl}
h_{1}:=&\lambda_{1}(2,2)(1-\sum_{k=0}^{1}\psi_{k})-a_{1}(2,2)\bar{a}_{2}(2,2)(1-\sum_{k=0}^{1}\psi_{k})+\sum_{k=0}^{1}\gamma_{k}\psi_{k},\vspace{2mm}\\
h_{2}:=&\lambda_{2}(2,2)(1-\sum_{k=0}^{1}\phi_{k})-a_{2}(2,2)\bar{a}_{1}(2,2)(1-\sum_{k=0}^{1}\phi_{k})+\sum_{k=0}^{1}\delta_{k}\phi_{k}.
\end{array}
\end{equation*}
Recall that $a_{1}(\underline{n})=a_{1}(2,2)$, $a_{2}(\underline{n})=a_{2}(2,2)$ for $\underline{n}\in S_{3}$.

In Figure \ref{stabi} we observe how the stability region is affected by varying $\lambda_{1}(\underline{n})$. In particular, when $\lambda_{1}(\underline{n})=0.2^{n_{1}}$, for $n_{1}=0,1,2$, $n_{2}=0,1$, the stability region is given by the triangular $ABC$. Note that the smaller value of  $\lambda_{1}(\underline{n})$ with respect to $\lambda_{2}(\underline{n})$, allows $\lambda_{1}(2,2)$ to take relatively large values with respect to $\lambda_{2}(2,2)$. When, we set $\lambda_{1}(\underline{n})=0.8^{n_{1}}$, the stability region becomes the triangular $DEF$, which seems to be more fair for $\lambda_{1}(2,2)$, $\lambda_{2}(2,2)$.
\begin{figure}
\centering
\includegraphics[scale=0.5]{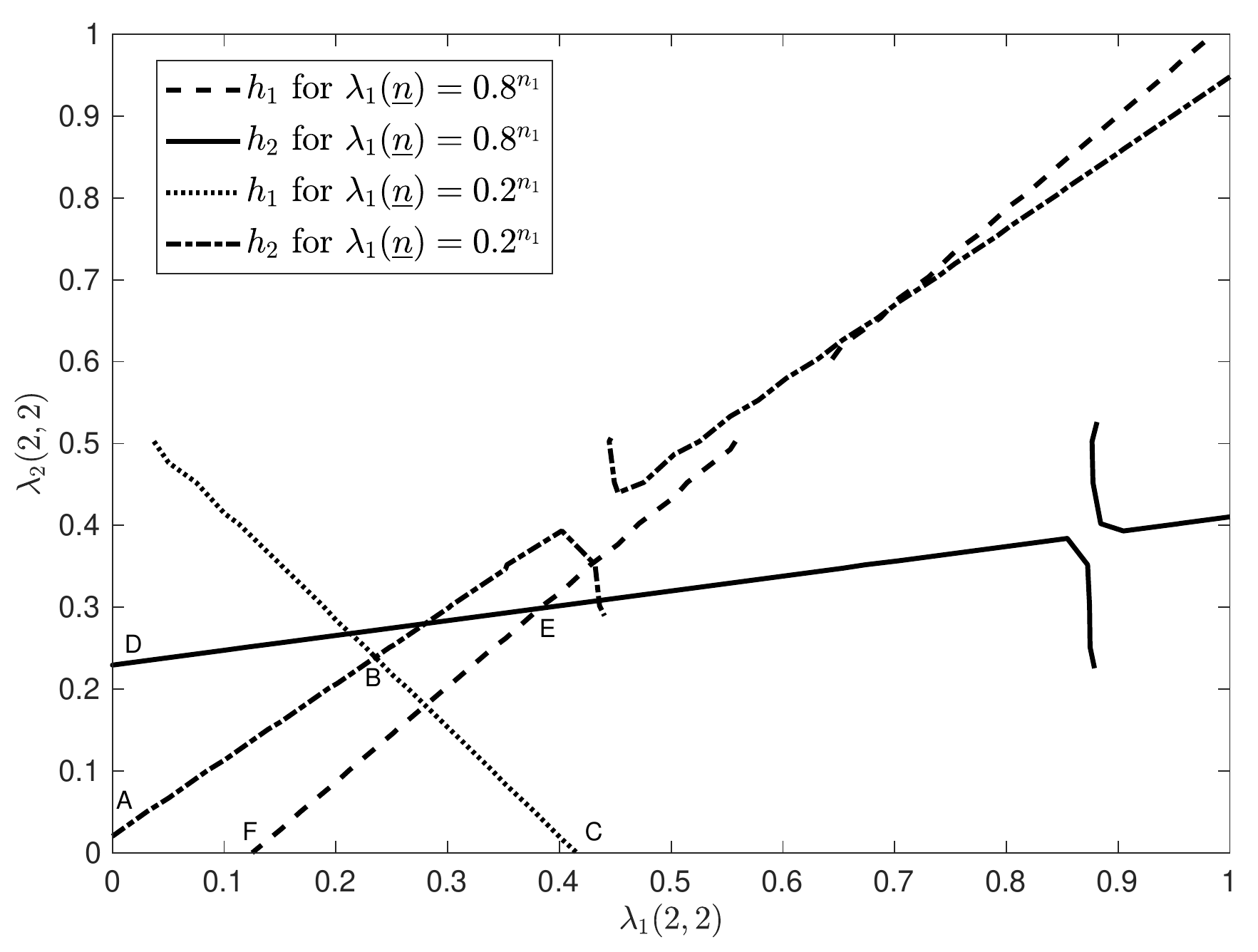}
\caption{The effect of $\lambda_{1}(\underline{n})$ on the stability region.}\label{stabi}
\end{figure}

\section{Conclusion}
In this work we provided an analytical approach to analyse the stationary behaviour of a partially homogeneous nearest-neighbour random walk in the quarter plane. We show that its stationary distribution is investigated by solving a functional equation using the theory of Riemann (-Hilbert) boundary value problem, along with a finite system of linear equations.

This class of random walks can be used to model plenty of practical applications including queue-aware multiple access systems. In such class of random access networks, intelligent nodes adapt their operational characteristics based on the status of the network. In a future work we plan to further investigate the numerical implementation of the approach as well as to compare it with other well known numerical oriented approaches such as the power series algorithm \cite{blanc1987numerical}.
\appendix
\section{Proofs for the case $\Psi(0,0)>0$}\label{nnrw}
\subsection{Proof of Theorem \ref{th1}}
It is readily seen that $\Psi(gs,gs^{-1})$ is for every fixed $|s|=1$ regular in $|g|<1$, continuous in $|g|\leq1$, and for $|g|=1$:
\begin{displaymath}
|\Psi(gs,gs^{-1})|\leq 1=|g^{2}|,
\end{displaymath}
and the proof of the first statement is a straightforward application of Rouch\'{e}'s theorem. Moreover, for $s=1$, by applying Rouch\'e's theorem in equation $g=g^{-1}R(g,g)$ it is seen that $g(1)=1$ is a zero of multiplicity one provided that $E_{x}<0$, $E_{y}<0$.
\section{Proofs for the case $\Psi(0,0)=0$}
\subsection{Proof of Lemma \ref{lemm1}}\label{ap}
For $|y|=1$, $y\neq1$, the kernel equation $R(x,y)=0$, or equivalently $xy=\Psi(x,y)$ has exactly one root $x=X_{0}(y)$ such that $|X_{0}(y)|<1$. This is immediately proven by realizing that $|\Psi(x,y)|<1=|xy|$ and applying Rouch\'e's theorem. For $y=1$, $R(x,1)=0$ implies $(1-x)[x(p_{1,0}+p_{1,1}+p_{1,-1})-(p_{-1,1}+p_{-1,0})]$. Thus, in case $\gamma:=p_{1,0}+p_{1,1}+p_{1,-1}-p_{-1,1}-p_{-1,0}<0$, $X_{0}(1)=1$. Similarly, we can prove that $R(x,y)=0$ has exactly one root $y=Y_{0}(x)$, such that $|Y_{0}(x)|\leq1$, for $|x|=1$. For an alternative derivation see \cite[Lemma 5.3.1]{fay}.
\subsection{Proof of Lemma \ref{lem1}}\label{ap1}
We will prove the part related to $\mathcal{L}$. Similarly, we can also prove part 2. For $x\in[x_{1},x_{2}]$, $D_{x}(x)=\widehat{b}^{2}(x)-4\widehat{a}(x)\widehat{c}(x)$ is negative, so $X_{0}(y)$ and $X_{1}(y)$ are complex conjugates. Thus, $|Y(x)|^{2}=\frac{\widehat{c}(x)}{\widehat{a}(x)}=k(x)$. Note that,
\begin{equation}
\frac{d}{dx}k(x)=\frac{x^{2}(p_{0,1}p_{1,-1}-p_{1,1}p_{0,-1})+2p_{1,-1}p_{-1,1}x+p_{-1,1}p_{0,-1}}{\widehat{a}(x)^{2}},
\end{equation}
where\footnote{To improve the readability we set $p_{i,j}:=p_{i,j}(N_{1},N_{2})=p_{i,j}(n_{1},n_{2})$ for $(n_{1},n_{2})\in S_{3}$.} $p_{0,1}p_{1,-1}-p_{1,1}p_{0,-1}=\lambda_{1}^{2}\lambda_{2}\bar{\lambda}_{2}\bar{a}_{1}a_{2}a_{1}\bar{a}_{2}>0$, and thus, $k(x)$ is a non-negative function for $x\in(0,\infty)$, which in turn implies that $k(x)\leq k(x_{2})$.

We can further solve $|y(x)|^2 = \widehat{c}(x)/\widehat{a}(x)$ as a function of $x$, and denote the solution that lies within $[x_1,x_2]$ by $\tilde{x}(y)$, i.e.,
\begin{equation}
\begin{array}{l}
\tilde{x}(y)=\frac{p_{0,-1}-p_{0,1}|y|^{2}-\sqrt{(p_{0,1}|y|^{2}-p_{0,-1})^{2}-4p_{-1,1}|y|^{2}(p_{1,1}|y|^{2}-p_{1,-1})}}{2(p_{1,1}|y|^{2}-p_{1,-1})}.
\end{array}
\label{cxz}
\end{equation}
So $\tilde{x}(y)$ is in fact the one-valued inverse function of $y(x)$. For each $y\in\mathcal{L}$ it also follows that
\begin{equation}
\begin{array}{l}
Re(y(x))=\frac{-\widehat{b}(\tilde{x}(y))}{2\widehat{a}(\tilde{x}(y))}.
\end{array}
\label{rd1}
\end{equation} 
Solving (\ref{rd1}) as a function of $|y(x)|^{2}$ then gives an expression for $|y(x)|^{2}$ in terms of $Re(y)$. 
\section{On the induced Markov chains in subsection \ref{sta}}\label{apc}
For $Q_{1,m}>N_{1}$, the component $Q_{2,m}$ evolves as a one-dimensional RW with one step transition probabilities $w^{(2)}_{j}(n_{2})=P(Q_{2,m+1}=n_{2}+j|Q_{2,m}=n_{2})$, $j=0,\pm1$, for $n_{2}=0,1,\ldots,$ given by
\begin{displaymath}
\begin{array}{rl}
w^{(2)}_{1}(n_{2})=&p_{0,1}(N_{1},n_{2})+p_{1,1}(N_{1},n_{2})+p_{-1,1}(N_{1},n_{2})=\lambda_{2}(N_{1},n_{2})[1-\bar{a}_{1}(N_{1},n_{2})a_{2}(N_{1},n_{2})],\\
w^{(2)}_{-1}(n_{2})=&p_{1,-1}(N_{1},n_{2})+p_{0,-1}(N_{1},n_{2})=\bar{\lambda}_{2}(N_{1},n_{2})\bar{a}_{1}(N_{1},n_{2})a_{2}(N_{1},n_{2}),\\
w^{(2)}_{0}(n_{2})=&p_{1,0}(N_{1},n_{2})+p_{-1,1}(N_{1},n_{2})+p_{0,0}(N_{1},n_{2}).
\end{array}
\end{displaymath}
Note that for $n_{2}\geq N_{2}$, $w^{(2)}_{j}(n_{2}):=w^{(2)}_{j}$, $j=0,\pm1$. Recall $\psi:=(\psi_{0},\psi_{1},\ldots)$ its stationary distribution. Then, simple calculations yields
\begin{displaymath}
\begin{array}{rl}
\psi_{n_{2}}=&\psi_{0}\prod_{j=0}^{n_{2}-1}\frac{w^{(2)}_{1}(j)}{w^{(2)}_{-1}(j+1)},\,1\leq j\leq N_{2},\\
\psi_{n_{2}}=&\psi_{N_{2}}\left(\frac{w^{(2)}_{1}}{w^{(2)}_{-1}}\right)^{n_{2}-N_{2}},\,j\geq N_{2}+1,\\
\psi_{0}=&[1+\sum_{n_{2}=1}^{N_{2}-1}\prod_{j=0}^{n_{2}-1}\frac{w^{(2)}_{1}(j)}{w^{(2)}_{-1}(j+1)}+\frac{w^{(2)}_{-1}}{w^{(2)}_{-1}-w^{(2)}(1)}\prod_{j=0}^{N_{2}-1}\frac{w^{(2)}_{1}(j)}{w^{(2)}_{-1}(j+1)}]^{-1},
\end{array}
\end{displaymath} 
where $w^{(2)}_{-1}-w^{(2)}(1)=\bar{a}_{1}a_{2}-\lambda_{2}<0$. Similarly, for $Q_{2,m}>N_{2}$, the component $Q_{1,m}$ evolves as a one-dimensional RW. Its one-step transition probabilities and its stationary behavior is derived as above and further details are omitted.

\bibliographystyle{spmpsci}      
\bibliography{bibl}   

%
%

\end{document}